\documentclass{amsart}
\usepackage{amsfonts,amssymb,graphicx}
\usepackage{epsfig,amsmath,latexsym}
\usepackage{amscd, amssymb, amsthm, amsmath,eucal, verbatim}
\usepackage{epstopdf}

\usepackage{amsmath}
\usepackage{algorithm}
\usepackage[noend]{algpseudocode}
\usepackage{caption}
\usepackage{subcaption}

\graphicspath{{../graphics/}{./graphics/}}

\usepackage[usenames,dvipsnames]{xcolor}

\newtheorem{theorem}{Theorem}[section]
\newtheorem{lemma}{Lemma}[section]
\newtheorem{proposition}{Propostion}[section]
\newtheorem{corollary}{Corollary}[section]
\newtheorem{claim}{Claim}[section]

\newcommand{\RR}{{\mathbb{R}}}


\begin{document}
\title{Approximating Surfaces in $\RR^3$ by Meshes with Guaranteed Regularity}
\author{Joel Hass and Maria Trnkova} 
\address{Department of Mathematics, University of California, Davis
California 95616}
\email{hass@math.ucdavis.edu, mtrnkova@math.ucdavis.edu}
\thanks{The first author was partially supported by NSF grants DMS1760485 and DMS1719582   }
\begin{abstract}
We study the problem of approximating a surface $F$  in $\RR^3 $ by a high 
quality mesh, a piecewise-flat triangulated surface whose
triangles are as close as possible to equilateral.  
The MidNormal algorithm generates a triangular mesh that is
 guaranteed to have angles in the interval  $[49.1^o, 81.8^o]$.
As the mesh size $e\rightarrow 0$, the mesh converges pointwise to $F$  
through surfaces that are isotopic to $F$.
The  GradNormal algorithm gives a piecewise-$C^1$ approximation of $F$,  with
angles in the interval   $[35.2^o, 101.5^o]$ as $e\rightarrow 0$.
Previously achieved angle bounds were in the interval $[30^o, 120^o]$.
\end{abstract}
 \maketitle

\section{Introduction} \label{intro}

The problem of finding a mesh, or surface made of flat triangles,  
that approximates a  smooth surface $F \subset \RR^3$ is important for a wide variety of
applications, including computer graphics, finite elements, finding numerical solutions of PDEs, and geometric modeling.  
A desirable feature in a mesh is the avoidance of ``slivers'', triangles that have one or more angles close to 
zero, which  can cause mesh-based algorithms  to break down
for numerical reasons.

The search for a mesh with good angle quality that approximates a given surface leads to two conflicting goals.
One goal is to make the triangles as nearly equilateral as possible, and the second is
to have the mesh converge smoothly  to the surface.  If one focuses entirely on angles,
one can construct a $C^0$  surface-approximation to $F$ consisting entirely of flat equilateral triangles,
as shown in Section~\ref{cubealg}.  The resulting mesh is unsatisfactory
in some ways, due to its normal vectors differing from those of $F$ by
more than $90^o$.  A new mesh giving a much improved $C^0$  approximation is described in  Section~\ref{theAlg}.
The  {\em MidNormal Algorithm} introduced there produces 
 a $C^0$ approximation of a  surface $F \subset \RR^3$ by a mesh 
with angles in the interval $[49.1^o, 81.8^o]$. These angle bounds are valid at any scale,
and  as the mesh size  approaches  zero it  converges  to  $F$ and is
 homeomorphic to $F$ under the nearest point projection. 
At the cost of  a less tight bound on the angles, 
we can get a $C^1$ approximation.
The  {\em GradNormal Algorithm} produces meshes with angles in the interval $[ 35.2^o, 101.5^o]$
that converge to $F$ piecewise-smoothly as the mesh size approaches zero.
For the GradNormal algorithm  the angle bounds are rigorously established for sufficiently fine meshes,
and can be compared to  the angle interval  $[ 30^o, 120^o]$ obtained by Chew's algorithm \cite{Chew93}. See also~\cite{ChengDey}.
Examples of the meshes produced by MidNormal are shown in Figure~\ref{MidNormalegs} and by GradNormal in Figure~\ref{GradNormalegs}.
The genus two surface is the level set   $((x^2 + y^2)^2 -x^2  + y^2)^2 +z^2  = 0.028 $.
The code used to produce these images can be found in \cite{NormalCode}.
The mesh is displayed using Meshlab \cite{MeshLab}.

\begin{figure}[htbp]
\centering
\begin{subfigure}{.33\textwidth}
  \centering
  \includegraphics[width=.9\linewidth]{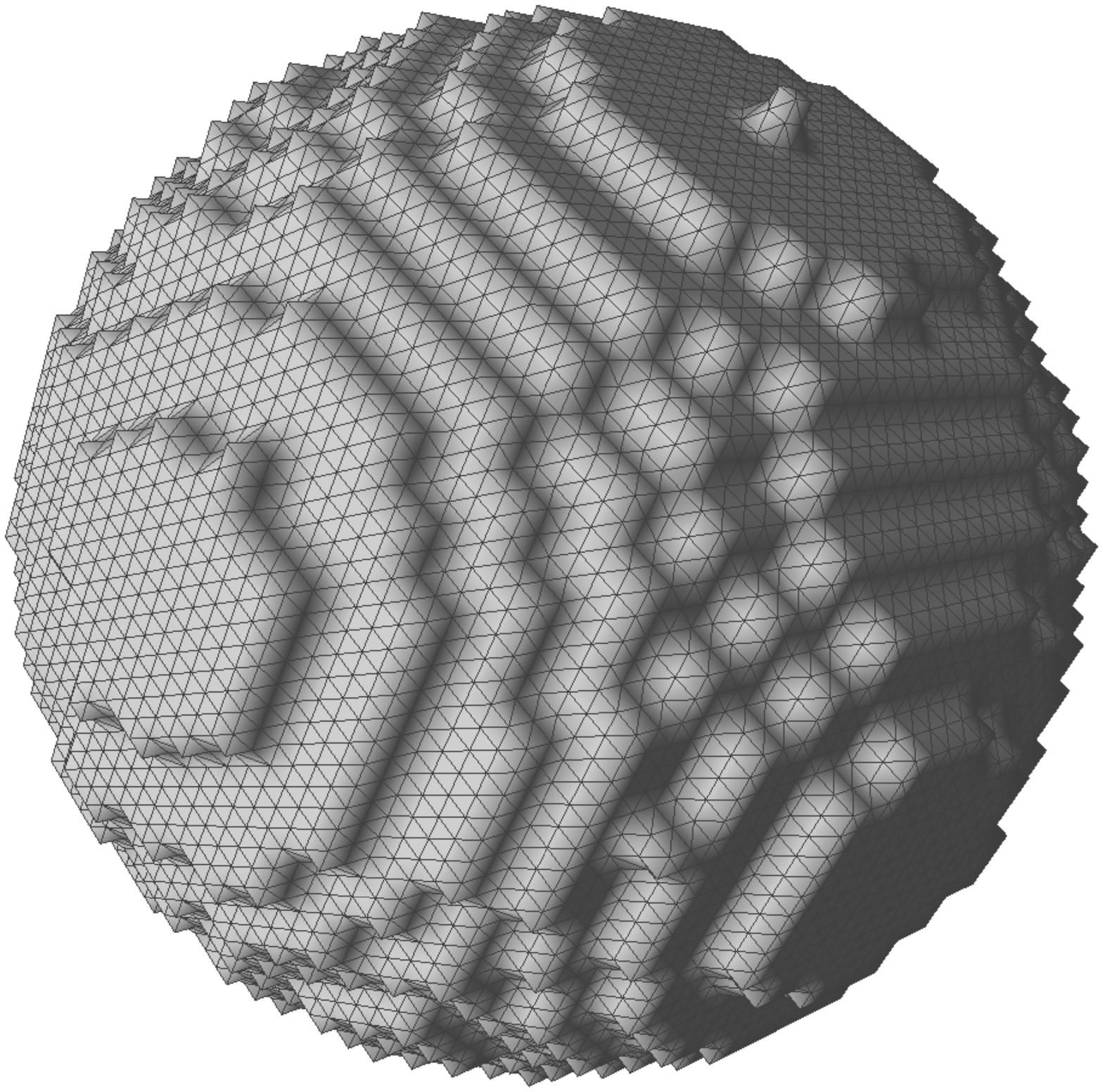}
  \caption{}
\end{subfigure}%
\begin{subfigure}{.33\textwidth}
  \centering
  \includegraphics[width=.91\linewidth]{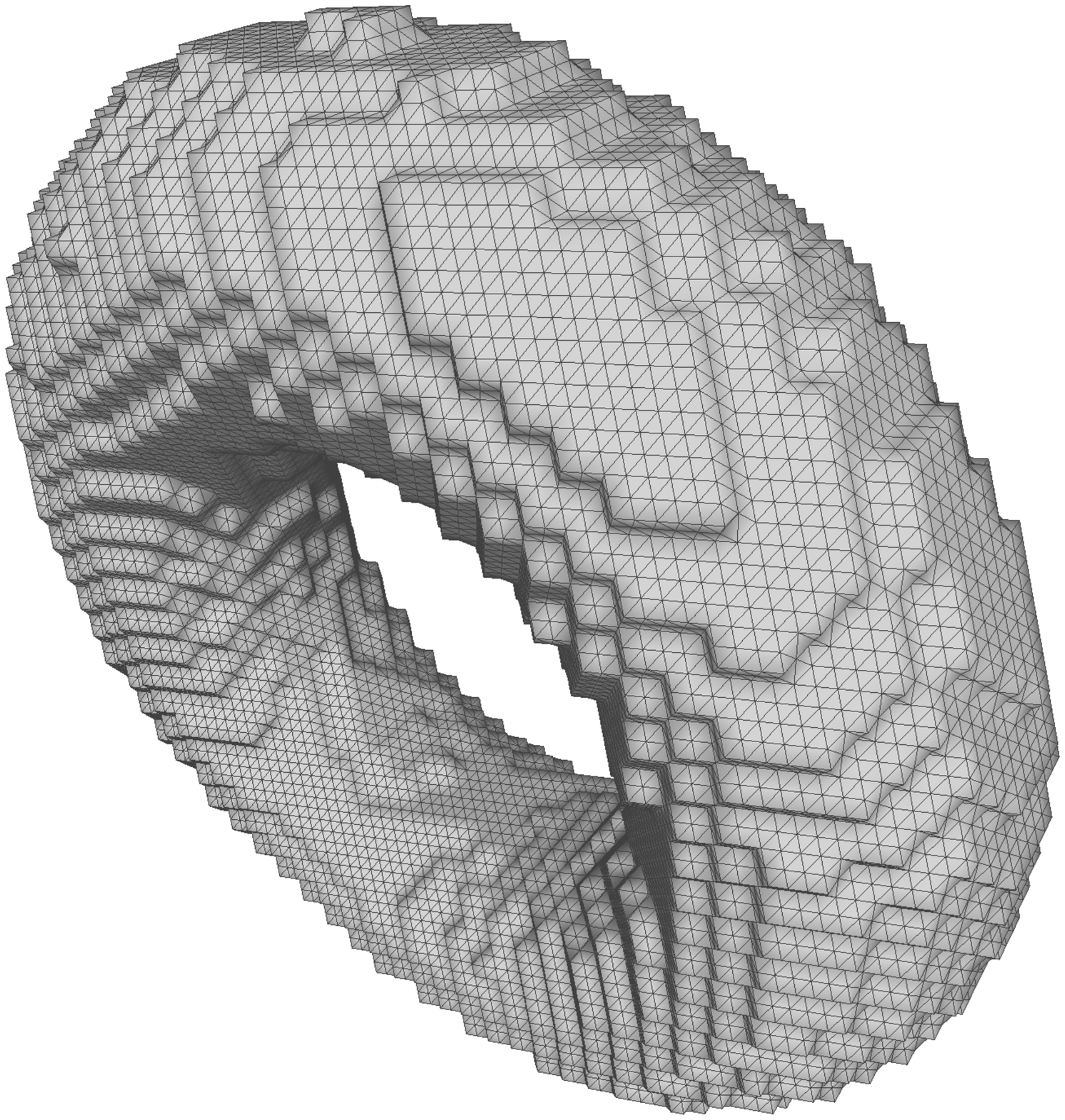}
  \caption{}
\end{subfigure}%
\begin{subfigure}{.33\textwidth}
  \centering
  \includegraphics[width=.8\linewidth]{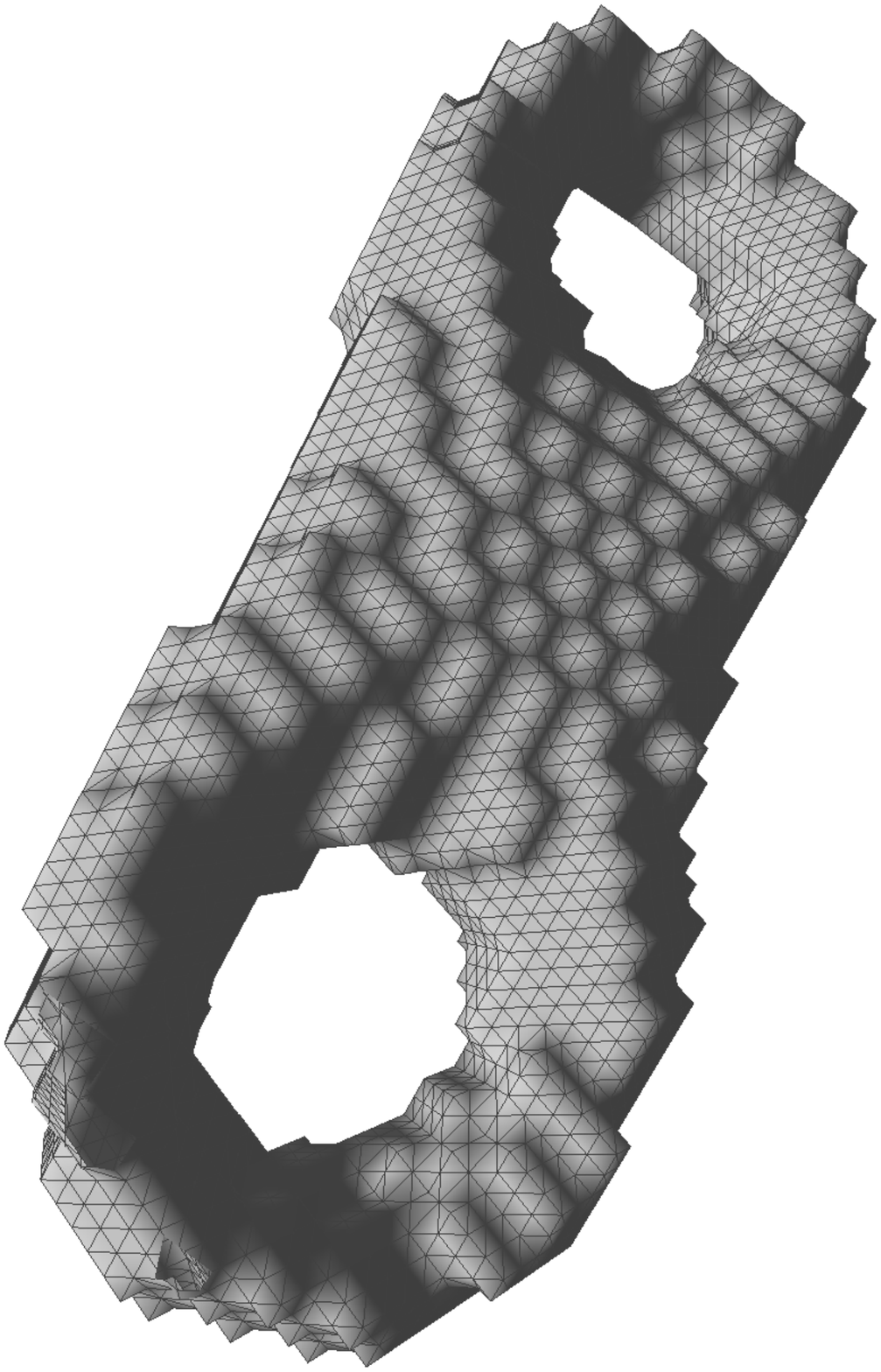}
  \caption{}
\end{subfigure}
\caption{ Meshes produced by the MidNormal algorithm. Tiling the unit cube with 709,920 tetrahedra gives meshes with  (A) 23,840  (sphere),  (B) 34,864  (torus) and
(C) 9968 (genus 2) triangular faces. These $C^0$-approximations have
all angles in the interval $[49.1^o, 81.8^o]$.} 
\label{MidNormalegs}
\end{figure}
  
\begin{figure}[htbp]
\centering
\begin{subfigure}{.33\textwidth}
  \centering
  \includegraphics[width=.9\linewidth]{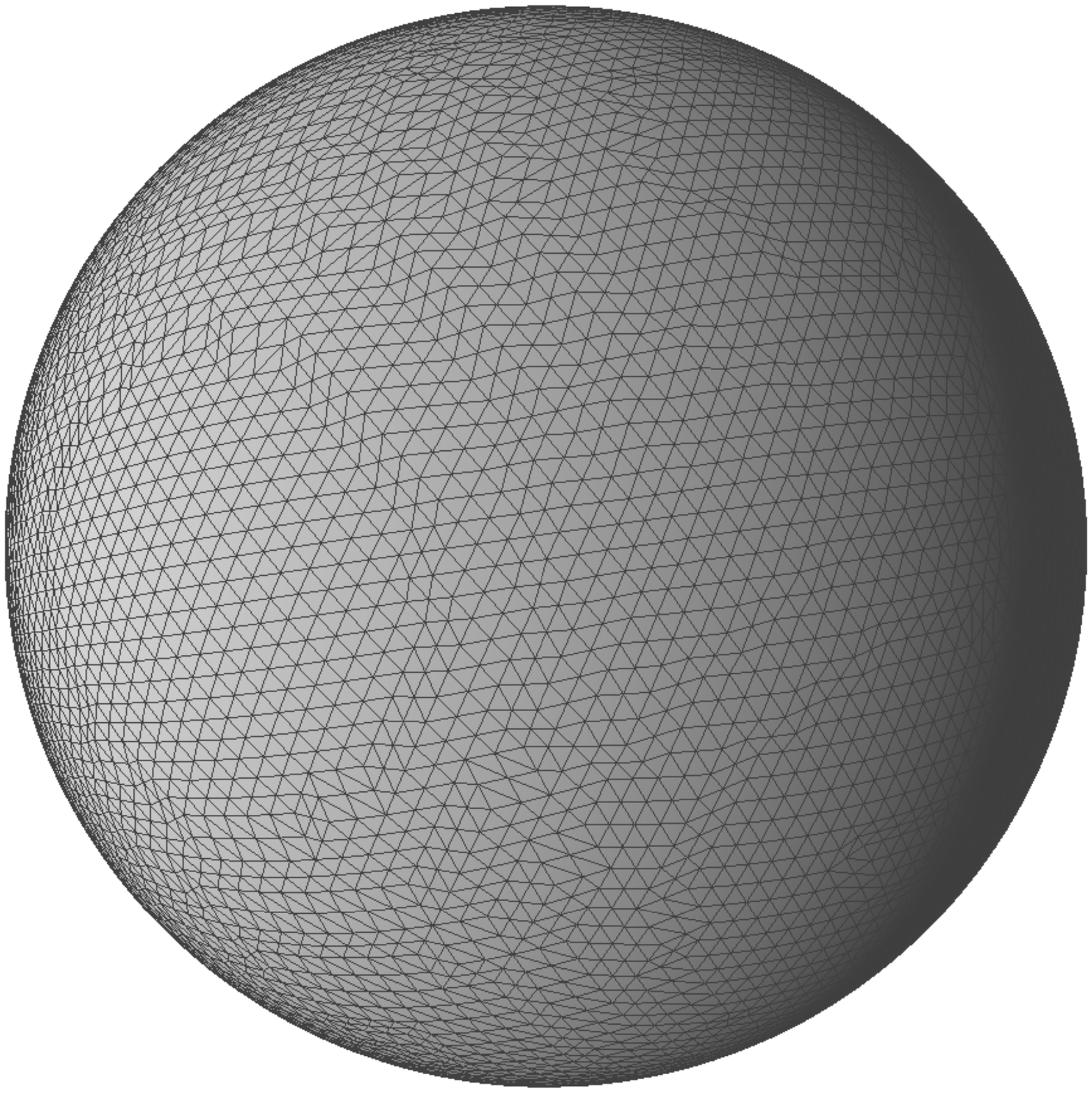}
  \caption{}
\end{subfigure}%
\begin{subfigure}{.33\textwidth}
  \centering
  \includegraphics[width=.75\linewidth]{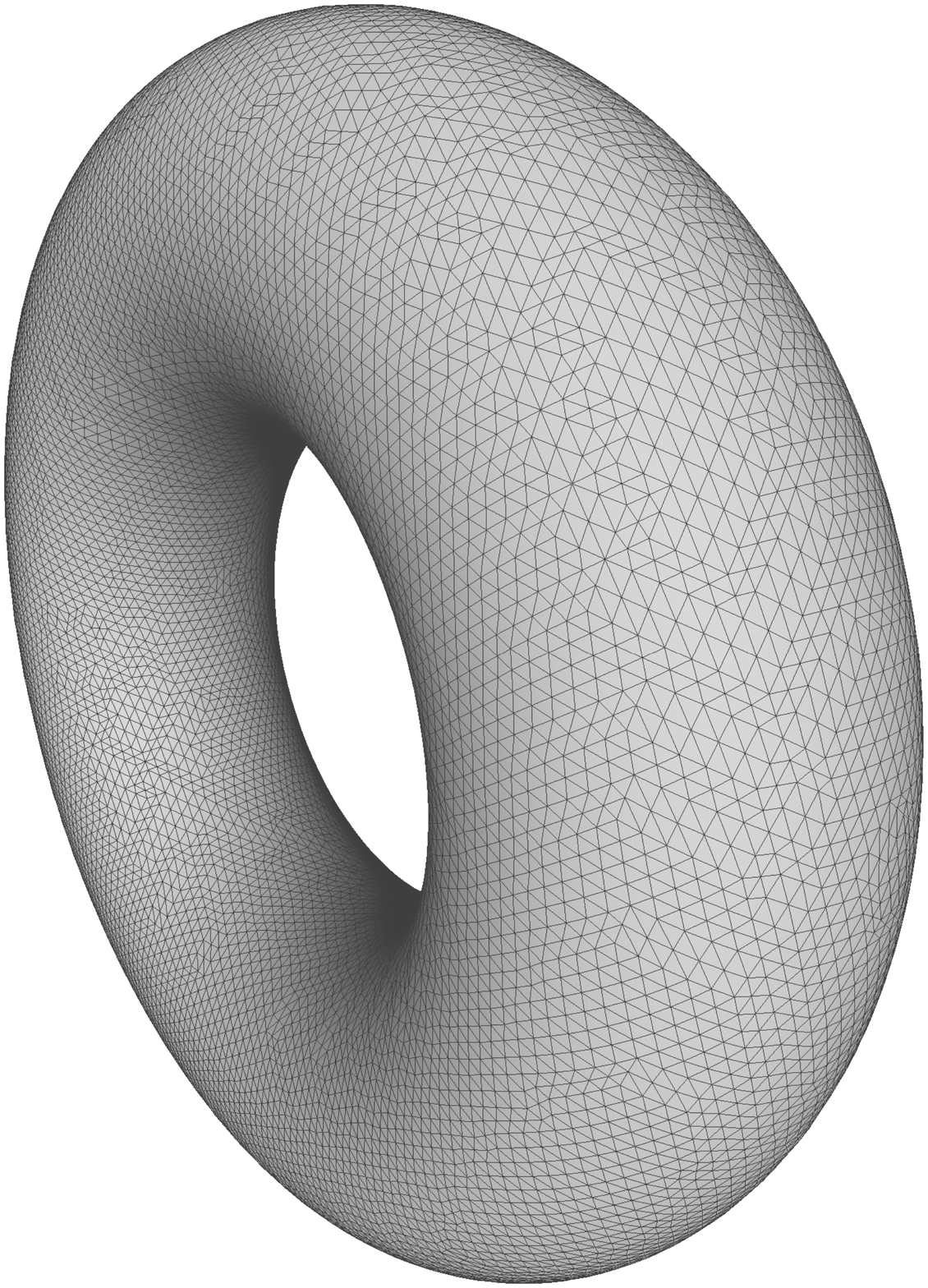}
  \caption{}
\end{subfigure}%
\begin{subfigure}{.33\textwidth}
  \centering
  \includegraphics[width=.8\linewidth]{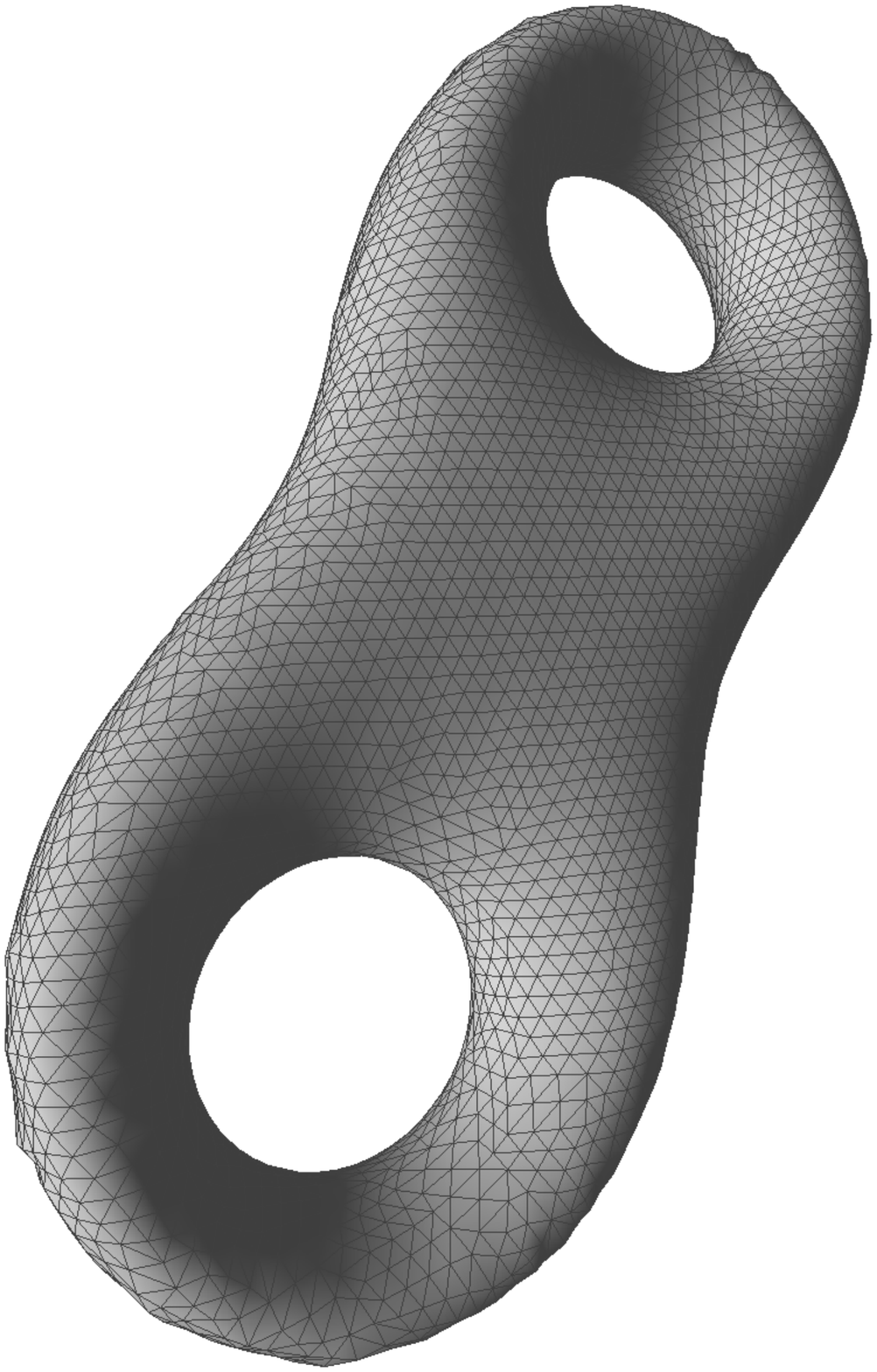}
  \caption{}
\end{subfigure}
\caption{Meshes  produced by the GradNormal algorithm. Tiling the unit cube with 869,652 tetrahedra 
gives meshes with (A) 25,092  (sphere), (B) 35,838 (torus)  and (C) 10,346 (genus 2)  triangular faces.
These  $C^1$-approximations have angles in the intervals
$[ 35.4^o, 102.7^o]$ (sphere),  $[ 32.8^o, 104.7^o]$ (torus) and $[ 22.1^o, 129.0^o]$ (genus 2). 
The greater curvature
of the higher genus surfaces results in lower mesh quality. 
At sufficiently fine scales all
 angles lie in the interval $[35.2^o, 101.5^o]$.} 
\label{GradNormalegs}
\end{figure}
 
The advantages of triangulations that avoid sliver triangles are discussed in the surveys
by Bern and Eppstein \cite{BernEppstein} and Zamfirescu \cite{Zamfirescu}.
A random process for selecting vertices on a surface gives a triangulation with expected
minimum angle approaching zero as the number of points increases \cite{BernEppsteinYao},
implying that  slivers are hard to avoid when creating meshes whose vertices come from  points sampled on a surface. 
Even more desirable, though difficult to achieve,
is a mesh in which all triangles are acute, so that every angle lying strictly between zero and 90 degrees.
Acute triangulations have been shown
to lead to faster convergence properties when used in various numerical methods \cite{Vavasis}.
Acute triangulations are also useful in establishing the convergence to a smooth solution
of discrete surface maps produced by computational algorithms that
approximate solutions to some classes of differential equations \cite{Luo}.

We  first investigate meshes that give a $C^0$ approximation to a surface.
This leads us to  the  {\em MidNormal Algorithm}, which produces 
meshes  with acute 
triangles having angles in the interval $[49.1^o, 81.8^o]$.
Moreover the sizes of the triangles are close to uniform across the entire surface. 
For a chosen scale constant $e$, the mesh gives edge lengths in the range $[e, 1.58e]$. 
The mesh produced by the MidNormal algorithm is guaranteed to describe an embedded 2-dimensional manifold, with no
gaps or intersecting triangles.
We show in Theorem~\ref{mainTheorem}  that as $e \to 0$,
the resulting meshes converge pointwise to  $F$, and the approximating
 mesh is isotopic to $F$ by the nearest point projection.

The idea of the MidNormal algorithm is to tile space with 
tetrahedra of a certain fixed shape and, given a surface $F$, to generate a triangulation  approximating $F$ by
intersecting $F$ with these tetrahedra.  This idea goes back to the theory of {\em normal surfaces}, a powerful tool used to
study surfaces in 3-dimensional manifolds that originated in work of Kneser in 1929 \cite{Kneser}.  The method was
first used for surface meshing  in 1991 in the marching tetrahedra algorithm \cite{DoiKoide, LorensenCline}. 
A mesh is generated by flat triangles that separate the vertices of a tetrahedron in the
same way that they are separated by $F$. The algorithm uses a tiling of $\RR^3$ by tetrahedra that 
belongs to a family of such tilings discovered by Goldberg \cite{Goldberg},
discussed  in Section~\ref{theAlg}. 
There is a family of Goldberg tetrahedra, each producing a tiling. A particular choice is determined by a pair of positive numbers, a shape parameter $a $ and a scale variable $e$. Together these determine a tetrahedron $\tau_{a,e}$ which we describe in
Section~\ref{theAlg}.   See Figure~\ref{GoldbergTet}.
Isometric copies of  $\tau_{a,e} $  fill $\RR^3$ with no gaps, matching along faces.  

The first version of the MidNormal algorithm uses a Goldberg tiling with parameter value $a_3 = \sqrt{3}/4 \approx 0.43$ and produces
angles in the interval  $[49.1^o, 81.8^o]$. This maximizes the minimum angle.
 A variation  of the algorithm  uses Goldberg tetrahedra of a second shape, $a_1 =(1/4)(\sqrt{ (19-3\sqrt{33}) /2   } ) \approx 0.2349$. 
This variation  realizes the global minimum for the upper angle bound and produces angles in the interval
$[38.3^o, 76.8^o]$.
Two other choices of $a$ give local extrema:  $a_2 =  1/\sqrt{11} \approx  0.30$ gives angles in the interval $[47.8^o, 84.2^o]$
and $a_4 =   \sqrt{ (3 \sqrt{17} -5 ) /32}\approx  0.47$ in the interval $[46.1^o, 77.3^o]$.

We show in Section~\ref{Comparison} that
the particular tilings we choose are close to optimal.
If we were to allow tetrahedra of different shapes in constructing a tiling,
or even if we use tetrahedra that overlap and fail to tile,
we could only  achieve a small improvement in the resulting angle bounds.

The  mesh properties given by the MidNormal algorithm are described in the following theorem.
The approximated surface $F$ is assumed to be given as the level set  $F = f^{-1}(0)$ of a smooth function
 $f: \RR^3 \to \RR$.  The method can be adopted to the case where $F$ is itself given as a mesh
 by taking $f$ to be the signed distance function from $F$.
In that setting, the algorithm can be used to replace a low quality mesh with one of higher quality.

\begin{theorem} \label{mainTheorem}  
Given a function $f: \RR^3 \to \RR$ for  which $F = f^{-1}(0)$ is a smooth compact surface with an embedded $\epsilon$-tubular neighborhood $N_\epsilon(F)$,
and a constant $e>0$, the MidNormal  algorithm produces a piecewise-flat triangulated surface $M(f,e)$ such that
\begin{enumerate}
\item  $M(f,e)$  is an embedded 2-dimensional surface. All triangle edges meet two triangles and
the link of every vertex is a closed curve.
\item Each triangle has angles in  the interval  $[49.1^o, 81.8^o]$.
\item Edge lengths fall into an interval $[e, 1.58e]$.
\item The triangles around a given vertex are graphs over a common plane. 
\item The surface $M(f,e)$  converges to $F$ in Hausdorff distance as $e \to 0$.
\item The surface $M(f,e)$ is isotopic to $F$  in $N_\epsilon(F)$ when $ e < \epsilon/ 2.$
\item The nearest neighbor projection from the mesh $M(f,e)$  to $F$  is a homeomorphism for $e$ sufficiently small.
\end{enumerate}
\end{theorem}

\begin{corollary} \label{converge} 
Given a function $f: \RR^3 \to \RR$ for which 0 is a regular value and $F = f^{-1}(0)$ is a compact surface,
the  MidNormal algorithm with scale $e$ produces  a triangulated surfaces $M(f,e)$ homeomorphic to $F$ with  angles in 
the interval  $[49.1^o, 81.8^o]$. As $e \to 0$ these surfaces
converge to $F$ in Hausdorff distance.
\end{corollary}

In an important set of applications, the function $f$  measures the
absorption at each point in $ \RR^3$ of an X-ray or imaging machine, or the density of a solid object.  
The desired level set $F = f^{-1}(0)$
can then represent the surface of a scanned object, such as an organ, bone, brain cortex or protein.
For purposes of visualization, geometric processing, surface comparison, surface classification,
or modeling of properties of the surface, it is desirable to have a high quality mesh representing the surface
such as that produced by MidNormal.

In many cases it is desirable to have a 
piecewise-$C^1$ approximation to the surfaces $F$.
We describe an algorithm that achieves this in Section~\ref{GradNormal}, which
we call the {\em GradNormal Algorithm}.  It starts with a mesh constructed using the MidNormal algorithm
with shape parameter $a_0= \sqrt{2}/4$. It then uses the gradient of $f$ to project
the vertices of the mesh towards the surface $F$.
Properties of the resulting mesh  $M^1(f,e)$ are 
described in the following result, proven in Section~\ref{GradNormal}.
 
\begin{theorem}
\label{GradNormalThm}
Let $F=f^{-1}(0) \in \RR^3$  be a compact level surface of a smooth function $f$. For sufficiently small scales $e$,\\
(1) The  triangular mesh  $M^1(f,e)$ is homeomorphic to $F$ by the nearest point projection to $F$.\\
(2)  As $e \to 0$ the surface  $M^1(f,e)$ piecewise-$C^1$ converges to $F$. \\
(3)  The mesh angles lie in the interval $[ 35.2^o, 101.5^o]$.
\end{theorem}

\subsection{Computational Methodology}  
Some of our arguments reduce to numerical  computations, which we carried out with the software package Mathematica 12.
Files are available on gitlab~\cite{NormalCode}. 
These computations operate with analytic functions, and can be carried out
with interval arithmetic to create completely rigorous arguments for the statements that we claim.
Real interval arithmetic is sufficient to give rigorous bounds for angle ranges, as required for our results.
Thus in principle all of our results can be made completely rigorous.  However we have not yet implemented
interval arithmetic into the numerical computations we used, so the actual values of the 
optimal angle bounds we give should be interpreted as numerically computed, and subject to the
correctness of the software being used.
This issue is discussed further when the computations are made.

The results in Section~\ref{Approximation} explore possible improvements of the results
of this paper by using similar techniques on other tetrahedral shapes. They rely on numerical computations
and are not stated as theorems.  They should rather be regarded as numerical evidence that
our techniques  cannot be improved substantially by using differently shaped tetrahedra.

The results in Section~\ref{SlidNormal} also rely on numerical computations.
The methods developed here not central to our methods, and are stated primarily to show that one  natural
approach does not give good results.
The angle bounds derived in this section rely on correctness of the numerical computations used, as is indicated there. 
The numerically  results in  Section~\ref{Approximation} and~\ref{SlidNormal}  are
not used  in the proofs of other results in this paper. 
 
\subsection{Normal Surfaces}
We start with a smooth function
$f: \RR^3 \to \RR$ that has 0 as a regular value, with the property that  $F = f^{-1}(0)$ is a smooth 
compact surface in $\RR^3$ that lies within the unit cube $[0,1]^3$.
We  describe an algorithm to generate 
a piecewise flat triangulated surface $M(f,e) \subset \RR^3$ that approximates $F$.
The algorithm will produce this mesh by assigning triangles  to some of the tetrahedra,
producing a {\em simple normal surface}.
 
A simple normal surface is a special case of the normal surfaces
introduced by Kneser in \cite{Kneser}.  Normal surfaces realize the simplest possible
way in which a surface can be situated relative to a three-dimensional triangulation, cutting
across each tetrahedron in the same way as a flat plane.
An  {\em elementary disk} in a tetrahedron is an embedded disk that is either a single flat 
triangle or two flat triangles meeting along a 
common edge and  forming a quadrilateral.
The vertices of each triangle in an elementary disk
are located at the midpoints of different edges 
of the tetrahedron.  See Figure~\ref{normal}.

\begin{figure}[hbtp]
\centering
\includegraphics[width=.6\textwidth]{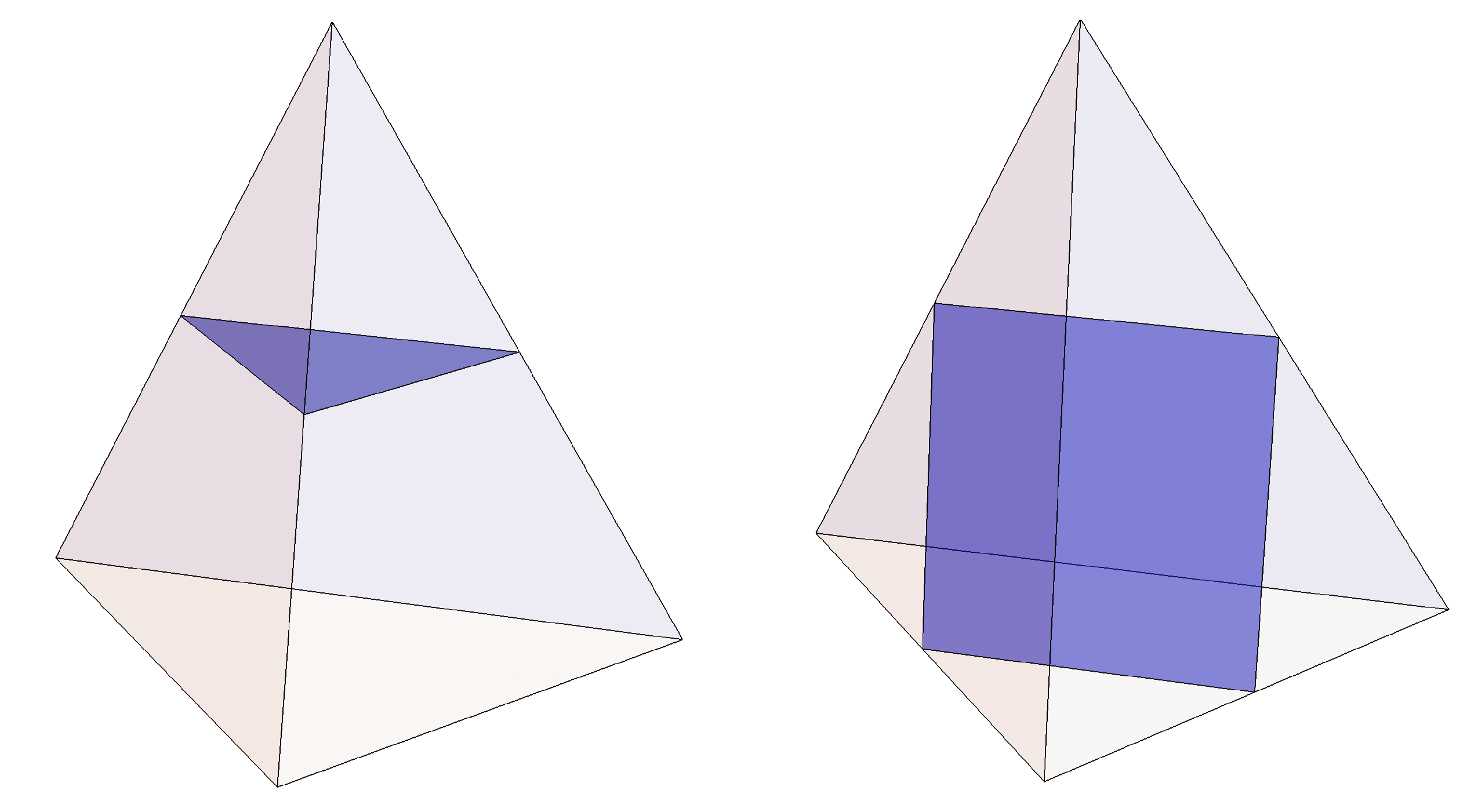}
\caption{
\label{normal}
Elementary disks forming part of a simple normal surface. A triangle separates one vertex from the other three.  A quadrilateral that separates
pairs of vertices is split into two triangles by adding a diagonal. There are four  possible triangles and three
possible quadrilaterals.}
\end{figure}

A general normal surface can intersect a single tetrahedron in many
parallel elementary disks.
By contrast, a simple normal surface with respect to 
a 3-dimensional triangulation $\tau$  is 
an embedded surface $S \subset M$  that
intersects any simplex in  $\tau$ transversely such that
the intersection of $S$ with any tetrahedron is either empty or
consists of a single elementary disk.
A simple normal surface can have more than one connected 
component, and may be non-orientable. 

%

The MidNormal algorithm, which we  describe in detail in Section~\ref{theAlg}, starts with   a surface in $\RR^3$ and 
produces a simple normal surface that
approximates this surface.
It produces triangles that are approximately of the same scale.
The edge lengths of any triangle fall into the range $[e, 1.58e]$, where $e$ is a constant that
can be set to any desired value.  As the edge lengths approach zero, the resulting surface
converges to $F$ in Hausdorff distance.  The faces of the approximating meshes have 
normal vectors  that lie in a fixed finite set of 18 normal directions, and therefore the approximation cannot be first order,
or piecewise-differentiable,
as the edge lengths approach zero. However the geometry of the
approximating surfaces is uniformly quasi-isometric to the limiting surface $F$.
We discuss methods of achieving a higher order approximation in Section~\ref{Approximation}.

\subsection{Related Work}  \label{related}

The MidNormal algorithm is closely related to the Marching Tetrahedra algorithm \cite{DoiKoide}, which is itself  a variation of 
the  Marching Cubes algorithm of Lorensen and Cline \cite{LorensenCline}. 
Marching Tetrahedra tiles space
with  tetrahedra obtained from subdivided cubes, and creates a mesh from triangles separating tetrahedral vertices in the same way as the surfaces we construct.  
It is widely used in areas such as medical imaging.
A drawback of Marching Tetrahedra is that the choice of cube based tilings can lead to low quality meshes.
The existence of sliver triangles is further exacerbated by adjustments to these algorithms that provide piecewise-$C^1$ approximations involving interpolating intersection points along edges.  The GradNormal
algorithm in contrast retains good angle bounds while producing a piecewise-$C^1$ approximation.
 
The problem of finding meshes with acute angles is difficult even for subregions of the plane with fixed boundary.
Algorithms exist to create Delaunay Triangulations for planar regions, which give various forms of
optimal regularity for a given vertex set \cite{ChengDey}.
However Delaunay Triangulations  are not in general acute, even for regions in the plane, and they can produce triangles with small angles.

It follows from work of Burago and Zalgaller that any polyhedral surface has a 
subdivision which is an acute triangulation  \cite{BuragoZalgaller}  (see also  \cite{Saraf}, \cite{ErtenUngor}).  
While it guarantees acute angles, 
 it gives no universal bound below $90^0$ for
the maximum angle that occurs.
Y. Colin de Verdiere and A. Marin showed that any smooth Riemannian surface admits a sequence of geodesic
triangulations with vertices on the surface and angles that, in the limit, lie in the intervals $[ 3\pi/10, 2\pi/5 ]$ for the case of genus zero, $[  \pi/3,  \pi/3 ]$ for  genus one, and
$[ 2\pi/7, 5\pi/14 ]$ for the case of genus greater than one  \cite{deVerdiereMarin}.
By Gauss Bonnet, these bounds are optimal for smooth surfaces.  These results are based on the Uniformization Theorem and
choosing an appropriate conformal model in the Moduli Space associated to the surface.
We are not aware of algorithms that construct
 meshes based on the above work.

Using different methods, Chew gave a procedure for mesh generation approximating a surface in  $\RR^3$ that
gives triangles that, while not acute,  have angles between $30^o $ and $ 120^o$  under certain assumptions \cite{Chew93}. 

Approaches to building surface meshes can be split into two groups. A {\em structured} mesh contains periodic combinatorics, and thus its triangle adjacencies
can be efficiently described.  An {\em unstructured} mesh has no particular global pattern, and is
therefore more adaptable to a variety of geometries
and topologies, but at a cost of requiring more bookkeeping to keep track of triangle adjacencies. 
The meshes produced by
our algorithms are structured in a 3-dimensional sense, similar to those produced by the
marching cubes and marching tetrahedra algorithms.  
They are induced by intersecting a surface with a triply periodic pattern of tetrahedra in $\RR^3$, whose
positions can be efficiently described by listing a triple of indices, but the resulting pattern of triangles does not need to have any 2-dimensional
periodicity. This captures some advantages of both the structured and unstructured approaches.

\section{A Pyramid Mesh Algorithm and its limitations}\label{cubealg}

In this section we show how to construct a mesh that is best possible if we consider only angles, and ignore issues
of tangent plane approximation.  We consider a surface-approximating mesh based  on triangles formed from the faces of
a cubical grid in $\RR^3$.  A collection of square faces
on the boundary of these cubes  limit to  $F$  in Hausdorff distance.
Subdividing each square into two  triangles by adding a diagonal gives 
a mesh with angles in the interval $[45^o , 90^o]$, but realizing only six normal directions.
If each square face is instead divided into four right triangles by adding a central vertex, the added vertex can be moved
orthogonally to the square to form a pyramid, decreasing the maximum angle in the mesh while avoiding self-intersections. 
The resulting mesh has angles in an interval
that is approximately $[55^o , 71^o]$. With additional care,  using an appropriate  subdivision and
choice of the two normal directions into which its vertex is moved, the pyramid
can be taken to consist of four equilateral triangles. 
The resulting mesh has best possible angles, all equal to $60^o$, and
all edge lengths equal.  

We state the  properties of this approximation in Theorem~\ref{pixels}
and then discuss its limitations.  Overcoming these shortcomings
is a main focus of this paper.

By scaling and translation, we can assume that $F$ lies in the unit cube $[0,1]^3$.  

\begin{theorem} \label{pixels} 
Let  $f: \RR^3 \to \RR$  be a function for which 0 is a regular value and for which 
$F = f^{-1}(0)$ is a closed surface contained in the unit cube $[0,1]^3$.
For a sufficiently large integer $N>0$ and a corresponding scale choice $e=1/N$, 
there is an approximating mesh $C(f,e)$ such that
\begin{enumerate}
\item  $C(f,e)$  represents an embedded 2-dimensional manifold.  
\item Each triangle in  $C(f,e)$  is equilateral, with all edge lengths equal to $e/6$.
\item  $C(f,e)$  converges to  $F$ in Hausdorff distance as $N \to \infty$.
\end{enumerate}
\end{theorem}
\begin{proof}
Tile the unit cube with  $N^3$ subcubes, each of side length $e$.
For $N$ sufficiently large, $F$ has an embedded tubular neighborhood of
radius $2e$.
Evaluate the function $f$ at all vertices of these sibcubes, and mark all cubes
that have both a negative and a
nonnegative vertex. Consider the region $B$ consisting of the union of these marked cubes.

Let $p \in F$ be any point in $F$.  
A radius  $e$ ball in $\RR^3$ tangent to $F$ at $p$ has interior that is disjoint from $F$.
If not, consider the point $q \in F$ closest to  the center of the ball $c$.  A normal line from $q$ to $c$ 
has length less than $e$, as does the normal from $p$ through $c$.  But we have assumed that 
$F$ has a tubular neighborhood of radius $2e$, which implies that any two normals of length 
$e$ are disjoint.  Thus $F$ is disjoint from any such ball.
So the two balls $B_1, B_2$  of radius  $e$ tangent to $F$ at $p$, one on each side of $F$,
have disjoint interiors and each meets $F$ only at  $p$.  Each subcube has 
diameter   $(\sqrt 3)e < 2e $. 
In particular,  the center point $c_1 \in B_1$ lies within
$e$  of some cube vertex $v_1 \in B_1$. Similarly the center point $c_2 \in B_2$ lies within
$e$  of some vertex  $v_2$ of a cube $B_2$ with $v_2 \in B_2$.
Moreover $f$ evaluates to a positive value on one of these vertices, say $v_1$ and a negative value on the other, $v_2$,
since they are separated by $F$.  The polygonal path from   $v_1$ to $p$ to $v_2$ 
has length less than $4e$ and  must cross a  cube in $B$.
Thus each point of  $F$ lies within distance $2e$ of a cube in $B$.

The  boundary of  $B$  is a 2-dimensional manifold, not necessarily connected, except for 
\begin{enumerate}
\item Isolated points meeting where exactly two cubes of $B$ meet at a common vertex,
\item A collection of edges that meet exactly
two cubes of $B$, and that  together form a graph.
\end{enumerate}

We now divide each $e \times e \times e$ cube of the tiling into 27 smaller cubes of side-length $e/3$,
and then add a layer of these smaller cubes to $B$ so that  $\partial B$ becomes a non-singular surface.
To do so, we add to $B$ all  $e/3 \times e/3 \times e/3$ cubes that meet the singular
set described above. The resulting collection of size $e/3 \times e/3 \times e/3$ cubes $C$
has a surface boundary $\partial C$ that is tiled by squares of size $e/3 \times e/3$.
Moreover the Hausdorff distance from $F$ to $\partial C$ is less than $2e$.

Next divide each  $e/3 \times e/3$ square  of $\partial C$    into four squares of
size $e/6 \times e/6$,
and further divide each $e/6 \times  e/6$ square into four triangles by adding a vertex at the center of the square,
and then displacing the added vertex perpendicularly from the square by a  distance of $\sqrt{2}e/2$, forming a pyramid
with equilateral faces. There is a choice of two normal directions in which to move the vertex,
and we now show how to choose this direction to ensure that two adjacent pyramids
do not intersect. 

Call a vertex of $\partial C$ {\em flat} if it has valence four and its four neighboring squares all lie in a plane.
The squares of $\partial C$  adjacent to a non-flat vertex
give a cycle in the link of that vertex in $\partial C$, namely a cycle in the 1-skeleton of an octahedron.
Thus vertices of $\partial C$ are adjacent to three, four, five or six squares. 
If the number is four or six, we alternate sides as we move around the vertex.
Thus  pyramids on adjacent squares deform into different sides and do not intersect.
If the number of squares around a vertex is three, then we orient them outward from the octant that they cut off.
If the number of squares around a vertex is five, then two of these squares are coplanar.
We orient the two pyramids on these squares in the same direction and alternate the direction of the
other three.   In each case the resulting pyramids have disjoint interiors.

Call the resulting triangulated surface $C(f,e)$.
Pyramids that have no common vertex have disjoint interiors, since the  distance of the base of a
pyramid from a  $e/3 \times e/3$ square of $\partial C$ that it does not intersect is at least $ e/3  >   \sqrt{2}e/12$.
Thus $C(f,e)$ is embedded, with the Hausdorff distance from $F$ to $C(f,e)$ is at most $2e$.
The claimed properties now follow.
\end{proof}

Note that we made no claim about the topology of the surface $C(f,e)$, which in fact will
not be homeomorphic to $F$ as constructed. In fact as constructed $C(f,e)$, 
will have more components then $F$. With more care we can arrange that $F$ is homeomorphic
to  $C(f,e)$, however we will not pursue this here.

The following {\em Pyramid Mesh Algorithm} constructs this mesh.
\begin{algorithm}
\caption{Pyramid Mesh Algorithm}
\begin{algorithmic}[1]

\Procedure{PyramidMesh}{$f,e$}       
\State Input;  A function $f: I^3 \to \RR$ and an edge length $e = 1/N$.
\For{$ i=1$ to $N^3$ }
\State Determine the sign of $f$  at each vertex of the 
$N^3$ sub-cubes $C_{i}$ of $I^3$.  
\State Create a list $B$ containing each sub-cube that contains both a negative and a positive vertex. 
\State Create a new list of cubes $B_1$  by dividing  each cube in $B$ into 27 sub-cubes  of  size $(e/3)^3$.
Add a layer of cubes around $B_1$ by adding to $B_1$ all cubes of size $(e/3)^3$ that  are not in $B_1$ but intersect a cube in $B_1$.
\State Divide each square in the boundary $\partial B_1$ into four sub-squares of size
$(e/6)^2$, producing a list of $M < 1296N^3$ squares $Q_{m}, ~~1 \le m \le M$.
\EndFor
\For{$ m=1$ to $M$ }
\State Add to a list of triangles ${\mathcal T}$
four equilateral triangles  $\{Q_{m} ^1, Q_{m} ^2, Q_{m} ^3,  Q_{m} ^4 \}$
forming a pyramid above each $Q_{i}.$ Orient this pyramid using the alternating rule described in Theorem~\ref{pixels}
that avoids interior intersections between distinct triangles. 
\EndFor
\State Output the list of equilateral triangles ${\mathcal T}$.
\EndProcedure
\end{algorithmic}
\end{algorithm}

\subsection{Limitations of the Pyramid Mesh Algorithm}
While the Pyramid Mesh Algorithm produces all equilateral triangles,
there are several issues 
that limit its applicability and lead us to
the normal surface based algorithms described in the following sections of this paper.
The first major drawback is that  the resulting triangles lie in planes
that do not give good approximations to nearby tangent planes of $F$.
The approximating mesh gives a zeroth order, but not a first order approximation of $F$.
One consequence of this
is that the induced metric on the approximating surfaces does not converge to that of 
the surface $F$.

A second limitation is that the nearest point projection from  the
 approximating mesh $C(f,e)$ to $F$ does not give a homeomorphism to $F$.  In fact, the nearest point projection
from $C(f,e)$ to $F$ can fail to be one-to-one, even as the mesh edge lengths $e$ approach zero.
This is illustrated in Figure~\ref{projectionproblem}.

\begin{figure}[htbp]
\begin{center}
\includegraphics[width=.3\textwidth]{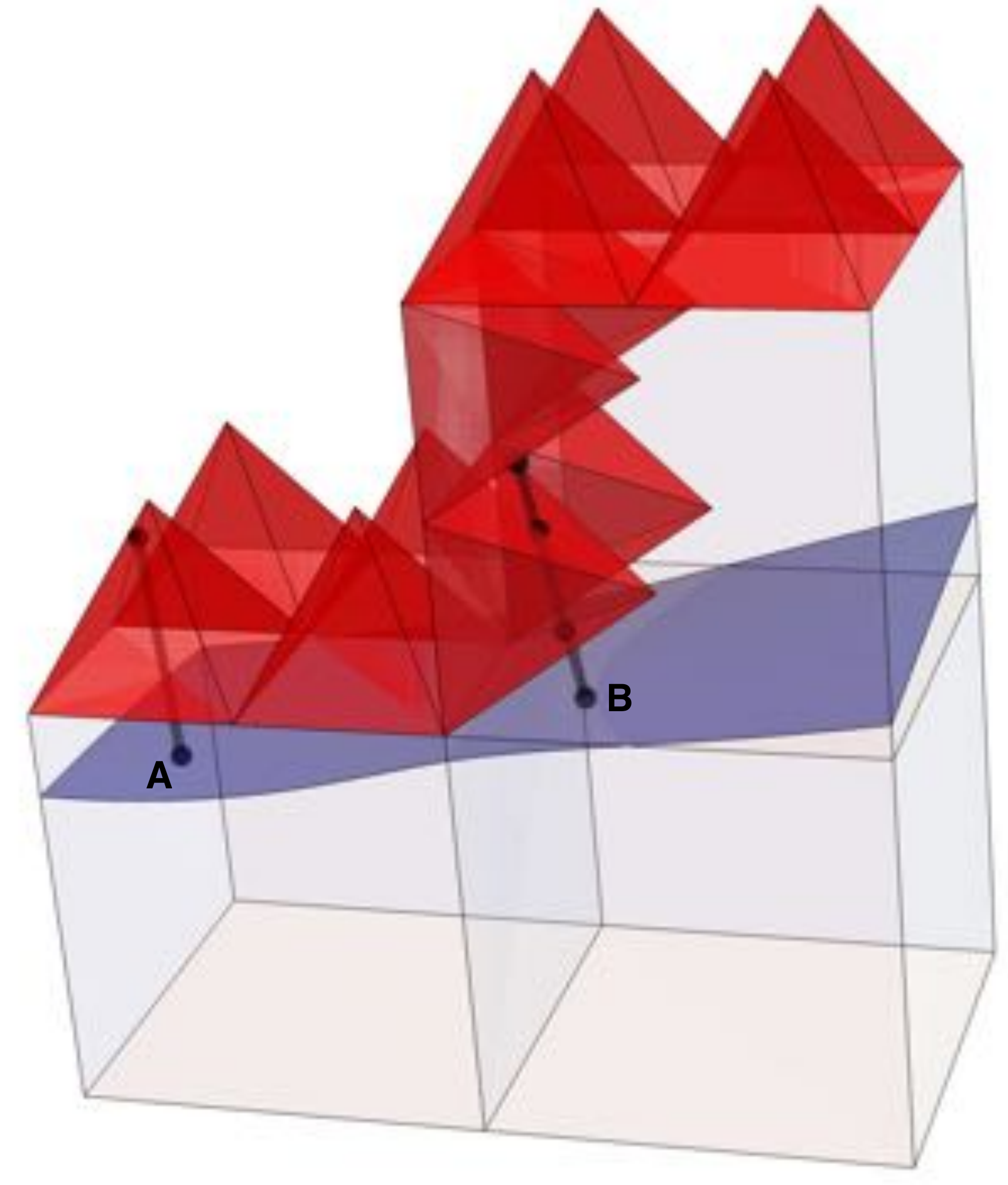}
\caption{A surface is represented by deforming the boundary of the collection of subdivided cubes so
that each square face is replaced by four pyramids, each with four equilateral triangular faces.
As indicated in the above example, the
 nearest point projection from the resulting equilateral mesh to $F$,
 represented here by a close-to-straight curve, may fail to give a homeomorphism,
even when the size of the cubes $\to 0$. The nearest point projection to point $B$ is not 1-1.} 
\label{projectionproblem}
\end{center}
\end{figure}

Another drawback is  that the Pyramid Mesh Algorithm yields extra components, which
would then need to be discarded.
Due to these limitations we will not explore it further in this paper.
 
\section{The MidNormal Algorithm} \label{theAlg}

In this section we describe in detail the MidNormal Algorithm, which generates a mesh from a tiling of $\RR^3$
by tetrahedra.  The algorithm takes as input a function $f: \RR^3 \to \RR$ with 
domain containing the unit cube $I^3 = [0,1] \times [0,1] \times [0,1] $ 
and outputs a mesh that approximates the level  set $F = f^{-1}(0)$.
It relies on filling $\RR^3$  with tetrahedra of a fixed shape and intersecting a surface with these tetrahedra to generate a triangulation.
In this section we describe for each constant $a>0$ a  tetrahedral shape $\tau_a$ that will be
used to tile $\RR^3$.  All the tetrahedra used for a given value of $a$ will be isometric (up to reflection).

It is natural to try to tile $\RR^3$ with near regular tetrahedra.
An interesting historical note is that
Aristotle claimed the false result that regular tetrahedra can meet five-to-an-edge and fit together to tile space \cite{Senechal}. 
In fact, the dihedral angle of a regular tetrahedron is somewhat less than  $2\pi/5 = 72^o$, so they don't
fit evenly around an edge.

The search for tetrahedra that do fit together  led Sommerville to find four tetrahedral shapes that tile $\RR^3$.  Baumgartner found a further example and 
Goldberg discovered three infinite families. Eppstein, Sullivan and Ungor constructed tilings of space by acute tetrahedra, with all dihedral angles less than $90^o$ \cite{EppsteinSullivanUngor}.
  
At first glance it may appeart that tetrahedra that are  as close to regular as possible  are preferable for producing
regular triangulations, but this is not the case. 
In fact, our method would not give meshes with acute triangles when applied to a regular or acute tetrahedron. 
The tilings that seems to work best for the MidNormal Algorithm are certain members of the family of tilings discovered by Goldberg.
A tetrahedron in the Goldberg family is shown in Figure~\ref{GoldbergTet}.
It is constructed by first tiling the $xy$-plane with equilateral triangles of unit length.
Three edges of the tetrahedron are graphs over edges of one of these equilateral triangles, each rising by 
a distance of $a$ from its initial to its final vertex. The other edges connect pairs of the resulting four
vertices.  The vertical edge $AB$ has length $3a$.  If we rescale by
a factor of $e$ then the equilateral triangle has edge length $e$ and the edge $AB$ has length
$3ae$.   The resulting tetrahedron is $\tau_{a,e}$. In scale independent computations we generally take $e=1$ for simplicity.
 
\begin{figure}[htbp]
\centering
\begin{subfigure}{.4\textwidth}
  \centering
  \includegraphics[width=1.1\linewidth]{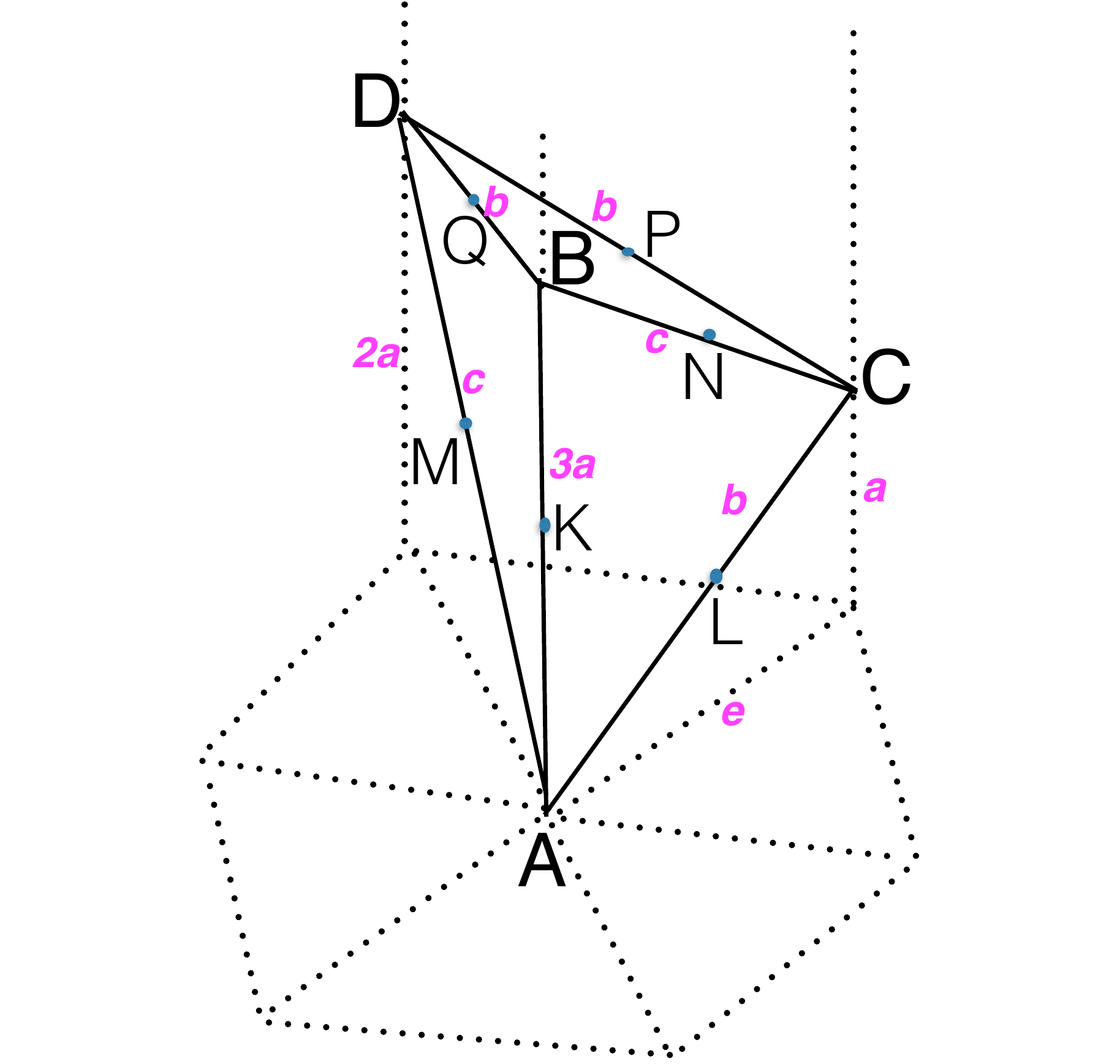}
  \caption{}
  \label{fig:sub1}
\end{subfigure}%
\begin{subfigure}{.2\textwidth}
  \centering
  \includegraphics[width=.8\linewidth]{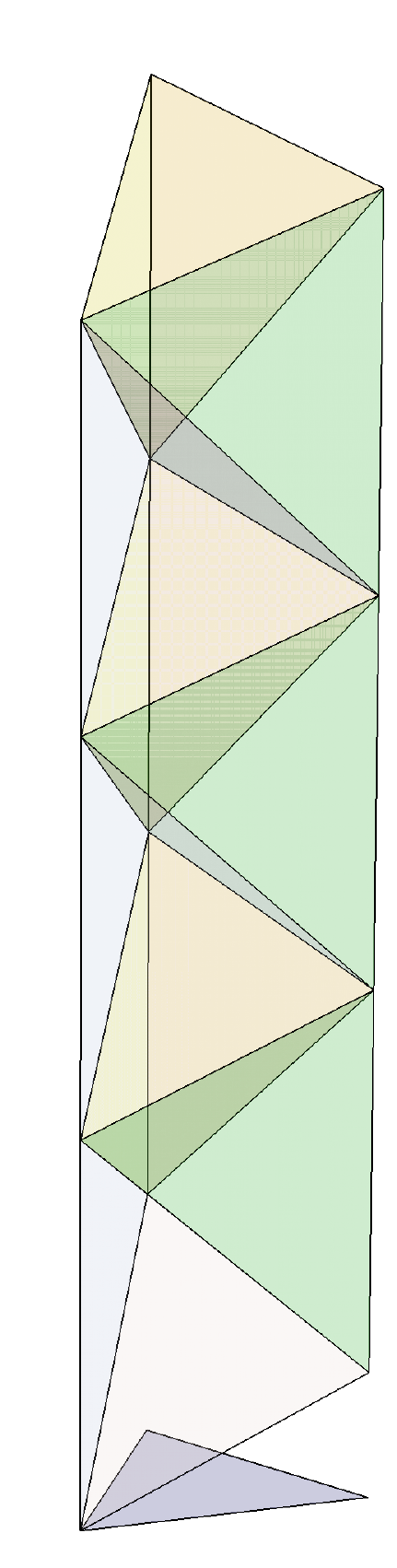}
  \caption{}
  \label{fig:sub2}
\end{subfigure}
~~~~
\begin{subfigure}{.325\textwidth}
  \centering
  \includegraphics[width=.9\linewidth]{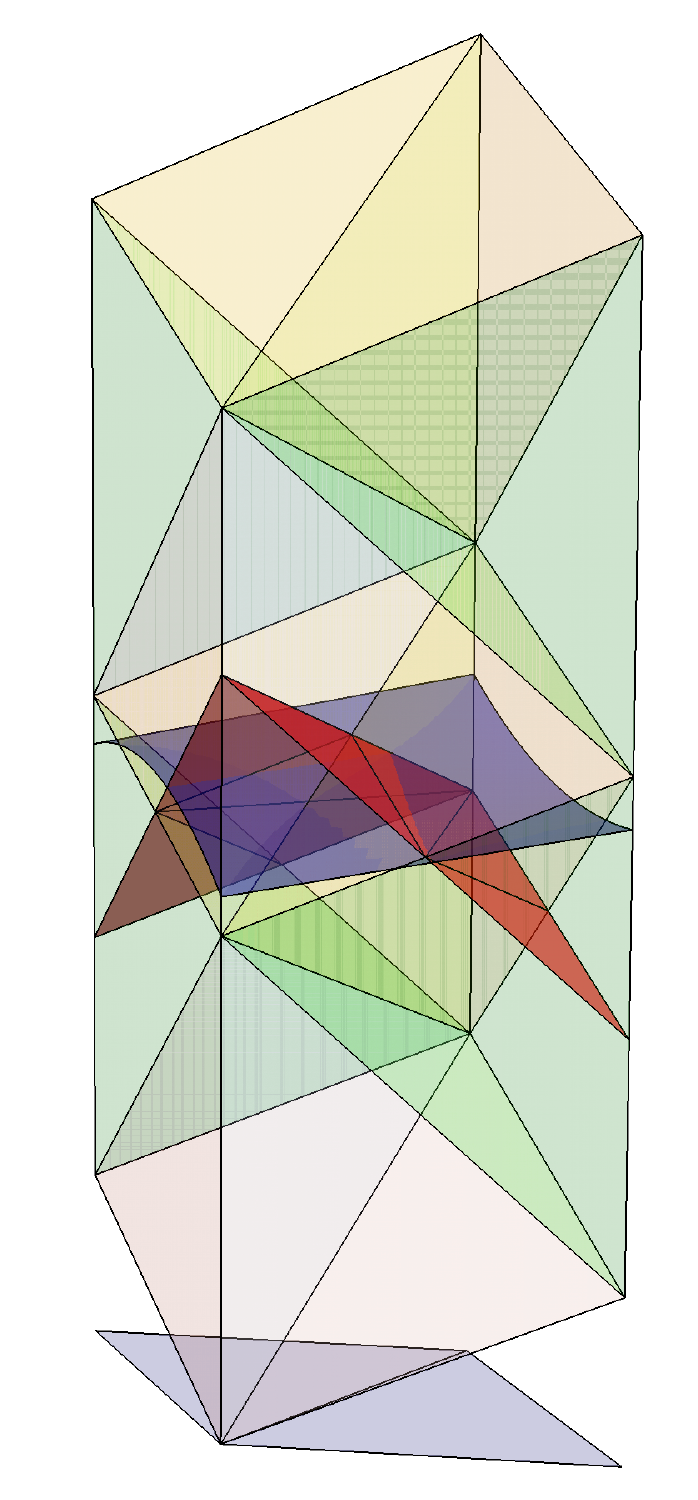}
  \caption{}
  \label{fig:sub3}
\end{subfigure}
\caption{(A) A tetrahedron $\tau_{a,e}$, one of a family  that tiles $\RR^3$. The scale independent 
parameter $ a \in (0,\infty )$ determines the
shape. (B) These tetrahedra stack to tile a vertical column over an equilateral triangle (center).  (C) A surface in $\RR^3$ divides
the vertices of these tetrahedra, leading to a triangular mesh. The tetrahedra shown here correspond to $a_3 = \sqrt{3}/4 \approx 0.43$.
Also shown in (C) is part of a normal surface.}
\label{GoldbergTet}
\end{figure}
 
We first prove a lemma establishing when a Goldberg tetrahedron is non-obtuse.
Certain tilings  from the Goldberg family  give a mesh with an acute triangulation.
We will compute the parameters $a$ that give optimal angles for $C^0$ and $C^1$ 
 approximations of a surface based on normal surfaces.
With the choice of tetrahedral tiling $\tau_a$ specified, we can describe the MidNormal algorithm.

\begin{algorithm}
\caption{MidNormal Algorithm}
\begin{algorithmic}[1]

\Procedure{MidNormal}{$a,e,f$}       
\State Input a function $f: I^3 \to \RR$, a choice of scale $e=1/N$, and a  shape parameter $a$ for
    a tiling of the unit cube $I^3$ by Goldberg tetrahedra $\tau_i$.
\For{$ i=1$ to $6N^3$ }
\State  Compute the sign of $f$ at the four vertices ${\tau_i}^1, {\tau_i}^2, {\tau_i}^3, {\tau_i}^4$ of tetrahedron $\tau_i$. If the value of $f$ at a vertex is exactly zero,  take the sign to be positive.
\State If the sign of $f$ is different at one vertex of $\tau_i$  from the sign at the remaining  three vertices, add to ${\mathcal T}$ the preferred triangle in $\tau_i$ that separates that vertex from the remaining three.
\State If the sign of $f$ is different at two vertices from the sign at the remaining  two vertices of  $\tau_i$, 
add to ${\mathcal T}$ two triangles formed by taking the quadrilateral in $\tau_i$ that separates the two pairs of vertices by adding a diagonal as follows:  
For quadrilateral $LMNQ$ add diagonal $MN$. 
For quadrilateral $KMPN$ add diagonal $MN$. 
For quadrilateral $KLPQ$ add diagonal   $LQ$   if $a \le a_0 = \sqrt{2}/4$ and diagonal $KP$ otherwise.
\EndFor
 \State Output the list of triangles ${\mathcal T}$.
\EndProcedure

\end{algorithmic}
\end{algorithm}

We will default the choice of $a$ to the value $a_3 = \sqrt{3}/4 \approx 0.43$, as we will show that this
gives a maximal lower bound on the mesh angle.
The value $a_0 = \sqrt{2}/4$ will be used to get a maximal lower bound on the the mesh angle in a second algorithm that gives  angle a $C^1$ approximation, as discussed in Section~\ref{GradNormal}. The value $a_0 = \sqrt{2}/4$ that appears in step (3)
corresponds to a case where the quadrilateral $KLPQ$ is a square,  and either choice of diagonal
gives rise to angles of  $90^o$ and $45^o$.
 
\subsection{Computations with Tetrahedra and Triangles} \label{computations}

We now compute the edge lengths and angles in the triangular meshes obtained by applying the MidNormal algorithm to 
the tetrahedral tilings of $\RR^3$ produced by the Goldberg family. We carry out a computation that shows how to choose
the parameter $a$ among this family of tilings to achieve optimal shapes for the triangles in these meshes.

The tetrahedron $ABCD(a,e)$ in Goldberg's family projects onto an equilateral triangle of side length $e$, as in Figure~\ref{GoldbergTet}.
It has edge lengths $(3a, b, b, c, c, c)$, where $b^2=a^2+e^2$ and $c^2=4a^2+e^2$.   
This tetrahedron and its reflection tile  $\RR^3$. The tiling is obtained by first taking copies of the
tetrahedron that are successively rotated by $2\pi/3$ horizontally and translated by distance $a$ vertically (in the $z$-direction). These fill up a vertical column 
above an equilateral triangle. Reflection through the faces of these vertical columns gives the tiling of $\RR^3$. 
We are interested in simple normal surfaces generated by these tilings.

These tiling tetrahedra are not acute, as the edges  $AD$ and $BC$  have valence four.
The size  $e$ of the equilateral triangles in the $xy$-plane gives a choice of scaling and
does not affect angles, so for simplicity we set $e=1$. 
The choice of the parameter $a$ determines the angles in the simple normal disks  through the midpoints 
 $K, L, M, N, P, Q$.
We now determine when a Goldberg tetrahedron is non-obtuse.

\begin{lemma}\label{nonobtuse}
The Goldberg tetrahedron $\tau_a$ with parameter $a$ is obtuse for $a > \sqrt{2}/2$ and non-obtuse for
$a \le \sqrt{2}/2.$
\end{lemma}
\begin{proof}
We compute dihedral angles from the dot products of normal vectors to faces of the tetrahedron $ABCD$. 
Normal vectors to faces $ABC$, $ABD$, $ACD$ and $BCD$ are: 
\[
(0, -3 a, 0), \quad
\frac{3a}{2}( -\sqrt{3},  1,  0), \quad 
\frac{\sqrt{3}}{2} ( a, \sqrt{3} a, -1 ), \quad 
\frac{\sqrt{3}}{2}( 2a, 0, 1 ).
\]
Of the six dihedral angles, three do not depend on the parameter $a$.
Edges  $AB$, $BC$ and $AD$ have dihedral angles $60^o, 90^o, 90^o$.   Edges $AC$ and $BD$ have dihedral angle
  $\arccos  (3a/ \sqrt{3+12a^2}) $, and are acute for all $a$. 
Edge $CD$ has dihedral angle $\arccos{(1-2a^2)/(1+4a^2)}$, and is non-obtuse for $a \in (0 , \sqrt{2}/2]$.
\end{proof}

Theorem~\ref{Range_of_a} characterizes which Goldberg tetrahedra have elementary normal
disks that are acute, taking into account all ways in which the three quadrilateral elementary disks
can be divided into two triangles by a diagonal.

\begin{theorem} \label{Range_of_a} 
The MidNormal algorithm gives rise to an acute  triangulation when  the parameter $a$ lies in
the  range $(0 ,\   \sqrt2/4 )\cup(  \sqrt2/4, \   \sqrt2/2)$. 
For $a = \sqrt2/4$ the algorithm gives a non-obtuse triangulation.
The value of $a$ that maximizes the minimum angle of the normal surface triangulation among the Goldberg tilings
is  
\[
a_3 = \sqrt{3}/4 \approx 0.43
\]
 and this gives a mesh with angles  in the interval $[49.1^o, 81.8^o]$.
This value of $a$ that minimizes the maximum angle among the Goldberg tilings
 is  
\[
a_1=\frac14\sqrt{\frac{19-3\sqrt{33}}{2}}\approx 0.2349,
\]
 and gives a mesh with angles  in the interval $[38.3^0 , 76.8^0]$.
 \end{theorem}  
\begin{proof} 
The tetrahedron $\tau_a$  from Goldberg's family with $e=1$ can be placed in $\mathbb R^3$ so that its vertices  $ABCD$  have coordinates 
$ A(0, 0, 0), B(0, 0, 3a), C(1, 0, a), D(1/2 , \sqrt{3}/2 , 2a). $
Its midpoints then have coordinates 
$  K(0, 0, 3a/2),  L(1/2, 0, a/2),    M(1/4, \sqrt{3}/4, a), \\  N(1/2, 0, 2a),   P(3/4, \sqrt{3}/4, 3a/2), Q(1/4, \sqrt{3}/4, 5a/2). $ 
The elementary normal disks that occur in the tetrahedron $ABCD$ are:\\
 Four triangles:  $\triangle KLM$ and $\triangle KNQ$ with edge lengths $(b/2, b/2, c/2)$,
and $\triangle LNP$ and $\triangle MPQ$ with edge lengths $( 3a/2, b/2, c/2)$.\\
Three quadrilaterals: $KLPQ$, $KMPN$ and $LMNQ$.

Each quadrilateral  can be split into pairs of congruent triangles using either of the two diagonals.  
The resulting edge lengths are 
given below:

\begin{enumerate} 
\item $KLPQ:$ In this quadrilateral  $KL=LP=PQ=KQ= c/2$ and the diagonal lengths are $KP= \sqrt3 /2 $, $LQ=\sqrt{4a^2+1/4}$. 
Dividing the quadrilateral into two triangles along each choice of diagonal gives
  \begin{enumerate}
  \item triangles $\triangle KLP\cong\triangle KPQ$ with edge lengths $ c/2, c/2, \sqrt{3}/2$, or
  \item  triangles $\triangle KLQ\cong\triangle LPQ$ with edge lengths $ c/2, c/2 , \sqrt{4a^2+1/4}$.
  \end{enumerate}
\item   $KMPN:$ In this quadrilateral $KM=MP=NP=KN=b/2$ and the diagonal lengths are $MN=c/2$, $KP= \sqrt3/2$. 
Dividing the quadrilateral into two triangles along each choice of diagonal gives
   \begin{enumerate}
   \item triangles $\triangle KMP\cong\triangle KPN$ with edge lengths $ b/2,  b/2,  \sqrt{3} /2 $, or
   \item  triangles $\triangle MKN\cong\triangle MPN$ with edge lengths $b/2, b/2, c/2$.
   \end{enumerate}
\item   $LMNQ:$  In this quadrilateral $LN=MQ= 3a/2$, $LM=NQ= b/2$ and the diagonal lengths $MN= c/2$, $LQ=\sqrt{4a^2+ 1/4}$. 
Dividing the quadrilateral into two triangles gives:
  \begin{enumerate}
   \item triangles $\triangle LNQ\cong\triangle QML$ with edge lengths $ 3a/2, b/2, \sqrt{4a^2+ 1/4}$,  or
   \item triangles $\triangle NLM\cong\triangle MQN$ with edge lengths $ 3a/2, b/2, c/2$.
   \end{enumerate} 
\end{enumerate}
We now consider how the  choice for the parameter $a$ affects the minimal and maximal angles in the mesh given by the corresponding simple normal surface.\medskip

In constructing a simple normal surface mesh we have a choice of which diagonal to add to decompose an  elementary disk of quadrilateral type
into two triangles. Elementary disks that are triangles may occur and cannot be avoided, so 
any of the triangles $\triangle KLM$, $\triangle KNQ$, $\triangle LNP$ and $\triangle MPQ$ might occur in a simple normal surface. 
Triangles congruent to these also appear in one of the two choices of diagonal subdivision for the quadrilaterals  $KMPN$ and $LMNQ$. 
Therefore for these quadrilaterals we do not need to consider the alternate choice of diagonal subdivision corresponding to
$\triangle KMP\cong\triangle KPN$ and $\triangle LNQ\cong\triangle QML$, as it would not lead to 
an overall improved angle bound. Choosing the diagonals that give
$\triangle MKN\cong\triangle MPN$ and $\triangle NLM\cong\triangle MQN$ will not affect the upper or lower bounds
for angles obtained anyway from the elementary disks that are triangles. 

Therefore a computation of the extremal mesh angles obtained from the tetrahedron $ABCD$ reduces to a
computation of the angles obtained for two triangles forming
the quadrilateral $KLPQ$ and for triangles  $\triangle KLM$ and $\triangle LNP$, as the parameter $a$ varies. 
 
\subsection{Optimal Shape Parameters}

We first examine the range of angle values  as a function of the shape parameter  $a$  for the congruent triangles 
$\triangle KLM$ and $\triangle KNQ$  and also for the congruent triangles 
$\triangle LNP$ and $\triangle MPQ$.   

\subsubsection{Triangle $\triangle KLM$:}

Note that edges $KM=ML=b/2 \leq KL=c/2$ depend on the parameter $a$.

By the  Law of Cosines
\begin{equation}\label{KLMboundaries}
\cos \angle KML =  \frac{1-2a^2}{2a^2+2}, \qquad \cos \angle KLM =  \cos \angle LKM =  \frac{\sqrt{4 a^2+1}}{2 \sqrt{a^2+1}}.
\end{equation}

\subsubsection{Triangle $\triangle LNP$:}

$LN=3a/2$, $NP=b/2 \leq LP=c/2$.
By the  Law of Cosines: 
\begin{equation}\label{LNPboundaries}
\cos \angle LPN =   \frac{1-2a^2}{\sqrt{a^2+1}\sqrt{4a^2+1}},  \quad \cos \angle NLP =  \frac{2 a}{\sqrt{4 a^2+1}},  \quad \cos \angle LNP =  \frac{a}{\sqrt{a^2+1}}. 
 \end{equation}

\subsection{Quadrilateral normal disks}
We now examine the range of angle values as a function of  $a$  for the 
triangles obtained by subdividing the three  quadrilateral elementary  normal disks into two triangles. For each quadrilateral
there are two choices of diagonal that divide it into a pair of triangles, and we examine each in turn.
As noted, the bounds obtained for the triangles also apply to a subdivision for quadrilaterals $KMPN$ and  $LMNQ$
and we need only compute which $a$ parameters give appropriate bounds for each of the two diagonal
subdivisions of quadrilateral $KLPQ$.
 
\subsubsection{Quadrilateral $KLPQ$}

The quadrilateral $KLPQ$ has four edges of length $c/2$, and its diagonals have length   $KP= \sqrt{3}/ 2 $ and $LQ= (1/2) \sqrt{16a^2+1}$.
There are two triangular subdivisions for  this quadrilateral, one for each diagonal. 
Each choice gives two congruent isosceles triangles: $\triangle KLP\cong\triangle KPQ$ and $\triangle KLQ \cong\triangle LPQ$,
and we consider each in turn.

\subsubsection{Triangle $\triangle KLP \subset KLPQ$:}  

$KL = LP =c/2$ and $KP = \sqrt3/2$ implies 
\begin{equation} \label{KLPboundaries}
\cos \angle KLP =\frac{2c^2-3}{2c^2}=\frac{8a^2-1}{8a^2+2}, \\
\cos \angle LKP = \cos \angle LPK = \frac{\sqrt{3}}{2c}=\frac{\sqrt3}{2\sqrt{4a^2+1}}.
\end{equation}

\medskip

\subsubsection{Triangle $\triangle KLQ\subset KLPQ$:}

$KL=KQ=c/2=\sqrt{a^2+1/4} \leq LQ=\sqrt{4a^2+1/4}$. 

By the  Law of Cosines
\begin{equation}\label{KLQboundaries}
 \cos \angle LKQ = \frac{1-8a^2}{8a^2+2}, \quad \cos \angle KLQ = \cos \angle KQL = \frac{\sqrt{16 a^2+1}}{2 \sqrt{4 a^2+1}}.
\end{equation}
 
\subsection{Summary}

The above conditions can be split into two cases: The angle bound are give by Equations (\ref{KLMboundaries}),
(\ref{LNPboundaries}), and (\ref{KLPboundaries})   or Equations (\ref{KLMboundaries}), (\ref{LNPboundaries}) and (\ref{KLQboundaries}),
depending on the choice of diagonal in quadrilateral $KLPQ$. 
The first case, illustrated in  Figure (\ref{fig:1A23}), 
gives acute values when $a\in ( \sqrt2/4 ,  \sqrt2/2) \approx (0.3535,  0.7071)$.  
The maximum of the minimal angle on this interval is realized at 
\[
a_3 = \sqrt{3}/4 \approx 0.43.
\]

With this choice for $a$, all angles of the triangulation are between  $49.1^0$ and $81.87^0$.
This achieves the largest minimum angle among this family of meshes.

Yet another choice, $a_4 =   \sqrt{ \frac{3 \sqrt{17} -5 } {32}}\approx  0.47$, minimizes the maximal angle for this choice of a diagonal of $KLPQ$ and gives angles in the interval $[46.1^o, 77.3^o]$.

\begin{figure}[htbp]
  \includegraphics[scale=0.7]{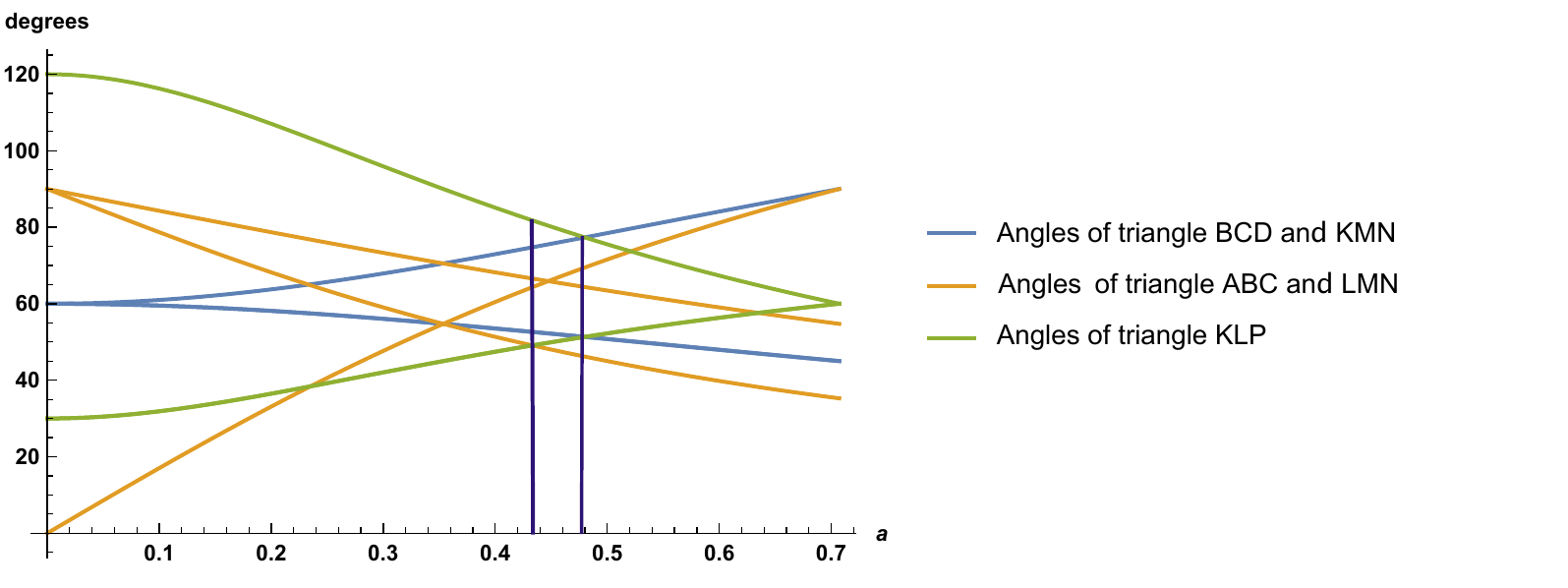}
  \caption{Angles of MidNormal triangles as functions of the tetrahedron shape parameter $a$. Quadrilateral KLPQ is split  along diagonal KP. }
  \label{fig:1A23}
\end{figure}

The second case returns acute triangulations  for $a\in (0 ,  \sqrt2 / 4 ) \approx (0 ,  0.353)$  as shown in Figure (\ref{fig:1B23}).
A minimum for the maximal angle on this interval occurs at
\[
a_1=\frac14\sqrt{\frac{19-3\sqrt{33}}{2}}\approx 0.2349.
\]

\begin{figure}[htbp]
  \includegraphics[scale=0.7]{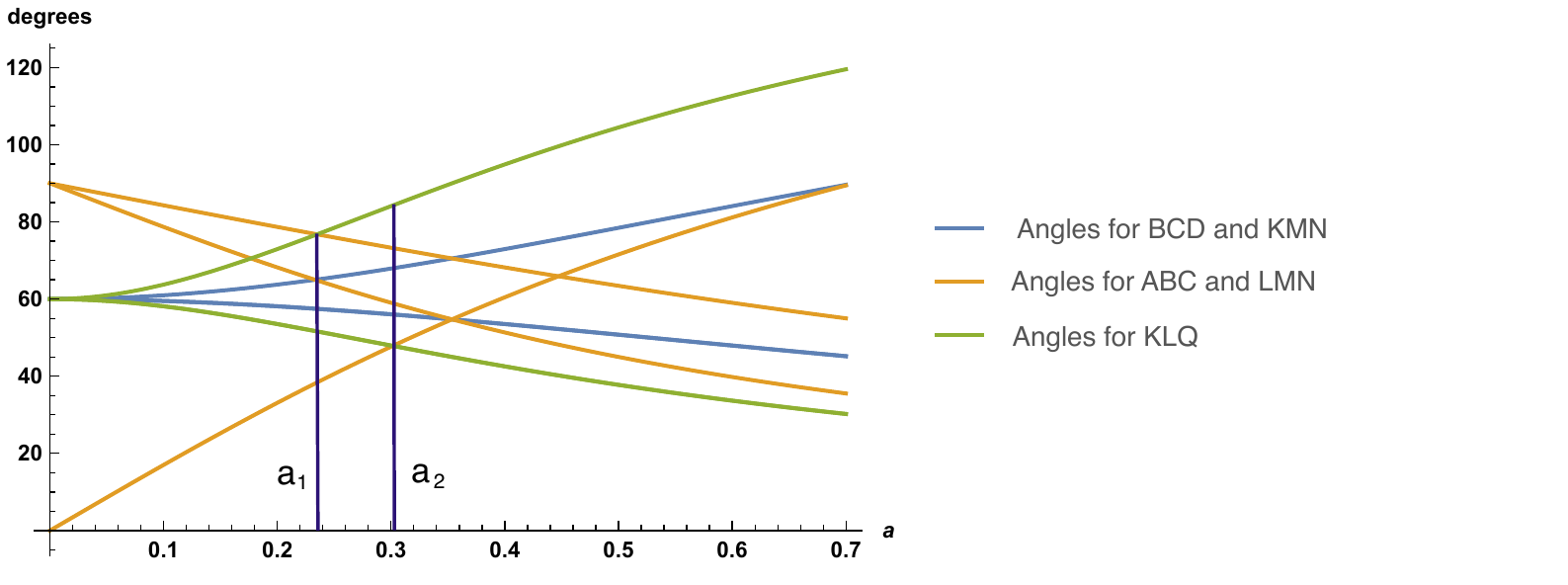}
  \caption{Angles of MidNormal triangles as functions of the parameter $a$. Quadrilateral KLPQ is split along diagonal LQ.}
  
  \label{fig:1B23}
\end{figure}
 
With this choice of diagonal and $a =a_1$,  all angles of the triangulation are between $38.3^0$ and $76.8^0$.
This achieves the smallest maximum angle for this family of meshes.
Maximizing the minimal  angle for this choice of a diagonal of $KLPQ$  gives  
 $a_2 =  \frac{1}{\sqrt{11}}\approx  0.30$,
and angles in the interval $[47.9^o, 84.3^o]$.

\end{proof}

\subsection{Convergence of the mesh to the surface}
 
We now analyze the mesh produced by the MidNormal algorithm and prove Theorem~\ref{mainTheorem}.
We first prove a lemma that relates the surface $F$ to the approximating mesh  $M(f,e)$.
\begin{lemma}  \label{$F$ to $M(f,e)$}
Let $F = f^{-1}(0)$ be a compact closed surface lying within a compact region
$B$ that is tiled by a family of tetrahedra $\tau_e$, in which each tetrahedron has diameter at most $d_e$.
Suppose that the surface $F$ has a tubular neighborhood of radius $\epsilon >0$
and that  $d_e \le \epsilon/2$. Then
the Hausdorff distance between $F$ and $M(f,e)$  is at most  $2d_e$ 
and  $M(f,e)$ is isotopic to $F$.  
\end{lemma}  
\begin{proof} 
Let $p \in F$ be any point in $F$.  
First note that a radius  $\epsilon/2$ ball in $\RR^3$ tangent to $F$ at $p$ has interior that is disjoint from $F$.
If not, consider the point $q \in F$ closest to  the center of the ball $c$.  A normal line from $q$ to $c$ 
has length less than $\epsilon/2$, as does the normal from $p$ through $c$.  But we have assumed that 
$F$ has a tubular neighborhood of radius $\epsilon$, which implies that any two normals of length 
$\epsilon$ are disjoint.  Thus $F$ is disjoint from any such ball.

It follows that the two balls $B_1, B_2$  of radius  $d_e$ tangent to $F$ at $p$, one on each side of $F$,
have disjoint interiors and that each meets $F$ only at  $p$.  By assumption  any tetrahedron in the tiling $\tau_e$ has 
diameter at most $d_e$ and any point in the tiled region $B$ lies within some tetrahedron,
so any point in $B$ is within $d_e$ of a vertex of $\tau_e$.  In particular,  the center point $c_1 \in B_1$ lies within
$d_e$  of some vertex  $v_1$ of a tetrahedron $\tau_1$ with $v_1 \in B_1$, and similarly the center point $c_2 \in B_1$ lies within
$d_e$  of some vertex  $v_2$ of a tetrahedron $\tau_2$ with $v_2 \in B_2$.
Note that $f$ evaluates to a positive value on one of these vertices, say $v_1$ and a negative value on the other, $v_2$,
since they are separated by $F$.
Now consider the piecewise-linear path $\alpha$ that starts at $v_1$, travels along a straight segment to $p$, and then 
follows a straight segment to $v_2$. By its construction, $M(f,e)$ separates the vertices of the tiling on which $f$ is positive
from those where it is negative. So $M(f,e)$ separates   $v_1$ from $v_2$ and  the path $c$ must cross
 $M(f,e)$ at least once.  Since all points  on the path $\alpha$ lie within distance $2d_e$ of the point $p$,
 the distance from $p$ to some point on $M(f,e)$ is at most $2d_e$.

Now consider an arbitrary point $q \in  M(f,e)$. Then $q \in \tau$ where $\tau$ is a tetrahedron in the tiling
and $f$ takes both positive and negative values on the vertices of $\tau$.
The surface $F$ must intersect  the same tetrahedron $\tau$, since by construction $F$ separates vertices of the triangulation
on which $f$ has different signs.
So the distance from  $q$ to $M(f,e)$ is at most  the diameter of the tetrahedron $\tau$, which is at most $d_e$.  
We conclude that the Hausdorff distance between $F$ and $M(f,e)$  is at most  $2d_e$.

We now look at the isotopy class of $M(f,e)$.  Note that $M(f,e)$ intersects the midpoint of each edge
that $F$ crosses an odd number of times.  Starting with $F$, we can do a series of normalization moves
that isotop it to $ M(f,e)$.  Each normalization move either removes a closed curve of intersection of
$F$ and a 2-simplex,  compresses $F$ inside a 3-simplex, or boundary compresses $F$ across an edge of
a tetrahedron, reducing the number of intersections with that edge by two \cite{Hass98}. 
When $F$ is incompressible and normalized in an irreducible manifold, one in which every 2-sphere is the  boundary of a
3-ball, then a compression will disconnect the surface, but one of the resulting components 
is a 2-sphere that bounds a 3-ball and  the other component is  isotopic to $F$.
Once all such moves are performed, the resulting surface is isotopic to a normal surface. 
Since none of these moves changes
which edges  the surface intersects an odd number of times,
the resulting normal surface is isotopic to $M(f,e)$.  

It remains to show that the resulting normal surface is isotopic to $F$.
In particular, we need to show that
any compressing move splits a trivial 2-sphere off  $F$ and preserves the isotopy class of the other remaining component. 
If $F$ is a sphere then this is immediate, so we assume $F$ has  genus at least one.

The tubular neighborhood  
$N_\epsilon(F)$ of $F$, which is homeomorphic to $F \times [-\epsilon, \epsilon]$,
contains all points whose distance from $F$ is at most $\epsilon$, and therefore
contains all tetrahedra that $F$ meets. 
The normalization procedure that carries $F$ to $M(f,e)$
introduces no intersections with tetrahedra that are not contained in $N_\epsilon(F)$. 
Thus the normalization isotopies and compressions take place completely within the
tetrahedra met by $F$.
Since $F$ is incompressible in $N_\epsilon(F)$ and $N_\epsilon(F)$ is irreducible when $F$ is not a 2-sphere, 
a compression gives rise to a surface isotopic to $F$, along with a trivial 2-sphere.
So the surface produced in the normalization process remains
in $N_\epsilon(F)$ and is isotopic to $F$ at each stage. We conclude that $F$ is isotopic to to $M(f,e)$, as claimed.
\end{proof}

We can now prove  the main result regarding the MidNormal algorithm.
\begin{proof} [Proof of Theorem ~\ref{mainTheorem}]
The MidNormal algorithm produces a triangulated surface that we call $M(f,e)$.
Since each edge meets exactly two triangles and each vertex meets four to six tetrahedra, and four to twelve
triangles arranged cyclically around the vertex, the mesh describes a 2-dimensional manifold.
The computation in Lemma~\ref{$F$ to $M(f,e)$} shows that the angles of the triangles for $\tau_a$ with $a_3 = \sqrt{3}/4 \approx 0.43$
lie in the interval  $[49.1^o,  81.8^o]$ and that the edge lengths fall into the range $[e, 1.58e]$.
Since all normals intersecting an edge $E \in \tau$ have positive component in the direction of the edge,
the triangles meeting $E$ are each graphs over the plane perpendicular to $E$. 
Every embedded surface has an $\epsilon$ neighborhood for some  $\epsilon>0$.
Lemma~\ref{$F$ to $M(f,e)$} implies that the mesh converges to the surface in Hausdorff distance as the edge lengths approach zero.
Note that $d_e \approx 1.41 e <2e$.  Finally note that as $e \to 0$ the smooth surface $F$ is increasingly well approximated by the
 tangent plane at a point of intersection of $F$ with the tetrahedron.
At a small enough scale the
projection from the elementary normal disk in the tetrahedron to $F$ is approximated by the projection
to this tangent plane, and is a homeomorphism. 
\end{proof}

\section{Approximation Quality} \label{Approximation}
The MidNormal algorithm gives a 0th-order approximation to a surface.
The faces of the approximating meshes produced by the algorithm have 
normal vectors  that lie in a finite set of directions, so one cannot expect to have
a first-order, or piecewise-$C^1$ approximation.
In general there are eight oriented normals arising from the four elementary triangles and up to 12 more from
the three elementary quadrilaterals, so up to 20 normal directions are possible with various choices of
how to subdivide a quadrilateral into two triangles.
Since the normals of the triangles in the mesh lie in a finite set,  the approximation cannot be piecewise-$C^1$,
even if the edge lengths approach zero.  
A computation shows that choosing $a = \sqrt{3}/4 \approx 0.43$ gives 18 oriented normals for
the midpoint triangles produced by the algorithm. This limited set of normal directions may sometimes be sufficient, but
for some applications it is desirable to get a piecewise-$C^1$-approximation, where the normals to the faces of the mesh  converge  to the surface normal as the mesh becomes sufficiently fine.  

The vertices of the mesh produced by the MidNormal algorithm
do not lie on $F$, the surface that is being approximated.
One way to achieve a piecewise-$C^1$-approximation is to move the vertices of the mesh to lie on $F$.
Note that each mesh vertex lies within distance $3ea/2$ of a point of $F$ that lies along an edge of the tetrahedra tiling.
When approximating a level surface $F= f^{-1}(0)$ by a mesh, we  can move each mesh vertex in the direction of $- \nabla f $ till it lands on $F$.
Each point moves a distance of less than $3ea/2 $.  

One strategy is to move vertices towards $F$ by relaxing the  condition that the mesh intersect
tetrahedral edges at midpoints, at the cost of getting weaker angle bounds on the mesh triangles. 
This leads to an algorithm, SlidNormal, in which vertices are slid along edges to improve the
first-order approximation.  A second approach is to project vertices onto $F$.  
This method can produce sliver triangles with arbitrarily small angles.  However we
will see that a modified projection, realized in what we call the GradNormal algorithm,  leads to a piecewise-$C^1$ approximation
with good mesh quality,
We  investigate this in Section~\ref{GradNormal}.

\subsection{Potential improvements from other tilings} \label{Comparison}

The MidNormal algorithm is based on tiling $\RR^3$ using one of  Goldberg's family of tetrahedra. A natural question is 
whether other tetrahedra that tile $\RR^3$ might give better quality for the resulting meshes. 
One difficulty in answering this is that it is not  at present known which single tetrahedral shapes can be used to tile $\RR^3$.
We nevertheless consider whether a tiling of $\RR^3$ by tetrahedra of some, perhaps unknown, shape, or 
even a tiling by multiple and varying shapes, might give better angles
for a normal surface mesh then the MidNormal algorithm based on Goldberg tetrahedra.

To investigate these questions we temporarily set aside the question of tiling and
just consider the angles of elementary disks inside a single tetrahedron of a given shape.
We search for two optimal tetrahedra. We first search for the  tetrahedron
that gives a mesh whose elementary disks have a smallest angle that is as large as possible and then 
 for the  tetrahedron
that gives elementary disks whose largest angle is as small as possible. In this search we
ignore the question of whether such a tetrahedron is part of a tiling of $\RR^3$.  We will see
that dropping the tiling condition does not give a significant improvement over the angles obtained
with Goldberg tetrahedra.  The Goldberg
tetrahedra, with $a$ appropriately chosen, give  close to optimal value for the
angles of its elementary triangles.

Let $ABCD$ be an arbitrary  tetrahedron in $\RR^3$.  After an isometry, scaling and relabeling of vertices, we can assume 
that:
\begin{enumerate}
 \item $AB$ is the longest edge,
 \item $A$ = (0, 0, 0), 
 \item $B= (1, 0, 0)$,
 \item $C = (x_C,y_C,0)$ lies in the $xy-$plane, 
and its coordinates satisfy 
\[
0 \le x_C \le 1,\quad 0 \le  y_C  \le  1, \quad \sqrt{1- {x_C}^2}\leq y_C\leq \sqrt{1-(x_C-1)^2}.
\]
 \item $D = (x_D,y_D,z_D)$ where
\[
1/2 \le x_D \le 1, \quad 0  \le y_D \le \sqrt{1-{x_D}^2}, \quad 0  \le  z_D\le \sqrt{1-{x_D}^2-{y_D}^2}.
\] 
\end{enumerate}
There is a five parameter space $(x_C, y_C, x_D, y_D, z_D)\subset \mathbb [0, 1]^5$ 
of possible tetrahedron shapes, and a subset $R$ of allowable values for
these parameters that satisfy the stated inequalities. 
There are seven elementary normal disks in a tetrahedron: four triangles and three quadrilaterals.  
Each quadrilateral can be triangulated in two ways, 
and this gives eight sets of triangulated elementary normal disks.  
We define eight functions $\gamma_i:\mathbb [0, 1]^5 \rightarrow \mathbb R, i\in \{1, \dots 8\}$,
each of which returns the minimal angle for one set of triangulated elementary normal disks in a  tetrahedron
of given shape.  
We consider each of these eight functions on the specified region $R\subset \mathbb [0, 1]^5$.
On this compact region in $\RR^5$ we  numerically search for a largest minimum angle. As described in Section~\ref{theAlg},  
a Goldberg  tetrahedron with parameter $a_3= \sqrt{3}/4 \approx 0.43$, 
has all mesh angles in the interval $[49.1^o,  81.8^o]$.
A computation  shows that only one  of these  functions $\gamma_i$
gives tetrahedra that have all elementary normal disks with  angles greater than $49.1^o$. 
This computation uses Mathematica to search the space  $\mathbb [0, 1]^5$ for a point  
which  maximizes the minimal angle of the elementary discs of a tetrahedron 
corresponding to this point in  $\mathbb [0, 1]^5$.  It is not relied on by
any other results in this paper.
 
It emerges that there is a single choice of triangular subdivision for each
quadrilateral elementary normal disk type that leads to values for the minimal
angles that are all larger than the $49.1^o$,
and that a maximum minimal angle  is realized by the $a = \sqrt{3}/4 \approx 0.43$ Goldberg tetrahedron.
This is represented by a tetrahedron with vertices at
$(0,0,0)$,  $(0,1,0)$,  $(0.38, -0.68, 0)$ and  $(0.57, 0, 0.67)$, 
and leads to angles in the interval $[49.69^o, 79.24^o]$.
This implies that the maximum smallest angle achievable over all $a$ values is $\approx 49.69^o$. 
This compares with the best angle of $ \approx 49.1^o$
obtained with  the  $a_3 = \sqrt{3}/4 $ Goldberg tetrahedron used in the midNormal algorithm.
Thus searching the entire space of tetrahedra, while ignoring tiling issues, increases the smallest angle bound by less than $1^o$, from $49.1^o$ to approximately $49.7^o$.
 
We next investigate the shape of a tetrahedron that minimizes the maximal angle when using  elementary normal disk triangulations. 
Three  of  eight choices of subdivision for the quadrilaterals lead to a better upper-bound on 
elementary normal disk angles than that given by  a Goldberg tetrahedron.
Recall that the parameter $a_1 =(1/4)(\sqrt{ (19-3\sqrt{33}) /2   } ) \approx 0.2349$ Goldberg tetrahedron $\tau_{a_1}$
gives rise to elementary normal disks that are triangulated with angles in the interval $[38.3^o, 76.8^o]$.
Three  tetrahedra, with appropriate choices for subdividing quadrilaterals, give angles in the ranges 
 $[45.29^o,73.29^o]$,  
$[15.01^o, 75.50^o]$, and
$[30.42^o, 72.92^o]$.

This analysis shows that searching the entire space of tetrahedral shapes, while ignoring tiling issues, leads
to a mesh with largest angle $ 72.92^o$,. This compares to the largest angle of $76.8^o$ achieved by 
 parameter  $a_1 $ in the midNormal algorithm.

In summary, numerical computation indicates that using tetrahedra of other shapes to 
tile space has the potential to improve the lower bound
on the angles produced by the MidNormal algorithm  by less than $1^o$, from about $49.1^o$ to  $49.7^o$,
and the upper bound by less than $4^o$, from  about $76.8^o$ to $73.0^o$.

\section{Sliding vertices} \label{SlidNormal}

The MidNormal algorithm gives a 0th-order approximation of a surface by a mesh whose vertices belong to a tetrahedral lattice. 
We describe below an approach to finding a  first order, or piecewise-$C^1$ approximation.  This
method turns out to  have limitations however, and we will instead develop for this purpose the better performing GradNormal algorithm described by in Section~\ref{GradNormal}.

An input  surface for the MidNormal algorithm is given as a level set. The algorithm constructs a tetrahedral lattice and evaluate the level function on its vertices. Then it takes a midpoint of an edge as a mesh vertex if the function changes its sign along that edge. 

The {\it SlidNormal} algorithm is a variation of the MidNormal algorithm. Instead of using a midpoint of a tetrahedral edge as a vertex of a mesh, we  slide this vertex along the edge to a location determined by a linear interpolation of the values of $f$ at the tetrahedra vertices. For example, let $AB$ be an edge and suppose a level function $\lambda: \mathbb R^3\rightarrow \mathbb R$  returns $\lambda(A)=-1$ and $\lambda(B)=3$. Then the linear approximation of a zero value of $\lambda$ is at a point $K\in AB$ such that $AK = (1/4)AB$.

\begin{algorithm}
\caption{SlidNormal Algorithm}
\begin{algorithmic}[1]

\Procedure{SlidNormal}{$e,f,a,t$}       
\State Input a differentiable function $f: I^3 \to \RR$ with
regular  level set $F = f^{-1}(0)$, a choice of scale $e=1/N$, a Goldberg tetrahedron shape parameter $a$, and $t\in [0,1/2]$.
\State  Apply the MidNormal algorithm with parameter $a$ to obtain a mesh $M(f,a,e)$ with $k$ vertices.
\For{$ i=1$ to $k$ }
\State Slide  vertex $v_i$ lying at the midpoint of edge $e_i$ towards the point on the edge where linear interpolation predicts $f=0$, but moving it at most $t$length$(e_i)$.
\EndFor
\State Output the list of triangles ${\mathcal T'}$.
\EndProcedure

\end{algorithmic}
\end{algorithm}

While sliding of vertices along edges brings them closer to the surface, the angles of the triangulation realize weaker lower and upper bounds. It is possible to obtain a triangle with very small or large angles (when  midpoints are moved  close to the vertices of the tetrahedra and the elementary normal disk limits to a tetrahedral edge).  

The next lemma gives some bounds on angles of a mesh produced by this SlidNormal algorithm. These bounds depend on a  sliding parameter $t\in(0,\,1/2)$. The parameter describes what portion of the edge length we allow vertices of normal disks  to slide along the tetrahedral edges.  Picking  $t=1/2$ allows the point to slide over the entire edge,  $t=0.25$ restricts the point to the middle half of the interval  and $t=0$ allows only the midpoint. 

The following lemma assumes correctness of a numerical computation, as we discuss in its proof.
We do not use these computations for our main results, which are independent of the results of this section.  

\begin{lemma} Suppose that we take 
a Goldberg tiling with  $ a_3 = \sqrt{3}/4 $ and produce a mesh with the SlidNormal algorithm using a sliding
parameter $t$.
Then any angle $\alpha$ of any triangle $\tau$ of the mesh produced satisfies:\\
$ 21.1^o  \le  \alpha  \le   116.2^o$ for $t = 0.2, $\\  
$ 16.1^o  \le  \alpha  \le   128.7^o$ for $t = 0.25$, \\
$ 11.9^o  \le  \alpha  \le  140.8^o $ for $t = 0.3$.
\end{lemma} 

\begin{proof} Let $ABCD$ be a Goldberg tetrahedron and let $K$, $L$, $M$, $N$, $P$, $Q$ be points on its six edges that depend on parameters $k, l, m, n, p, q \in (1/2-t,\,1/2+t)$. There are four triangles and three quadrilateral normal disks  produced by these six points. Each quadrilateral can be triangulated in two different ways. Again, as in the previous section, there are 8 different cases of triangulating normal disks. 
For each case we use Mathematica to numerically compute the minimal angle of each triangle of 
all elementary normal disks as a parameter of the sliding distance $t$. This gives a function  of  $k, l, m, n, p, q$ whose minimum depends on $t$. 
We numerically find the minimum value of this function over the region $[1/2-t,\,1/2+t]^6$.

We then choose how to triangulate quadrilaterals among the 8 cases so as to get the largest lower bound. 
For example,  for $t = 0.3$ the minimal angle arising for all normal disks is $11.9^o$. 
We then check the maximum of the largest angle among all triangles when  vertices slide in the same interval 
$(1/2-t,\,1/2+t)$. For example, for $t = 0.3$ the maximal angle among all triangles is $140.8^o$.  
A similar computation applies for other values of $t$.
\end{proof}

If we change the tiling parameter to $a_1 =(1/4)(\sqrt{ (19-3\sqrt{33}) /2   } ) \approx 0.2349$  then the results are slightly worse:\\
$ 18.9^o  \le  \alpha  \le   116.8^o$ for $t = 0.2 $,\\  
$ 13.22^o  \le  \alpha  \le   129.2^o$ for $t = 0.25$, \\
$ 10.4^o  \le  \alpha  \le  155^o $ for $t = 0.3 $.

We conclude that sliding vertices along edges of the tetrahedra to allow for
a closer to $C^1$ approximation has a significant cost in mesh quality.  An alternate approach, described in
Section~\ref{GradNormal}, gives much better results, both in the accuracy of the approximation and
in the quality of the mesh.

\section{The GradNormal algorithm} \label{GradNormal}

In this section we discuss the GradNormal algorithm,  an exension of the MidNormal algorithm
that gives a piecewise-$C^1$ approximation of a level surface. 
This overcomes the limitations associated with the limited sets of tangent planes
of the MidNormal algorithm and the poor angle quality obtained by sliding vertices along edges as in
Section~\ref{SlidNormal}.
It produces triangles that, under appropriate curvature assumptions on $F$ or mesh scaling assumptions, are
contained in the interval $[ 35.2^o, 101.5^o]$. 
This represents a significant improvement compared to the best previous bounds on angles for a piecewise-$C^1$ approximation ,
which were  obtained by
Chew  \cite{Chew93}, and gave angles in the interval  $[ 30^o, 120^o]$. 

This section involves computations  that compute derivatives of derivatives of explicit functions,
and also estimates of functions of one variable along a closed interval.
The angle values we obtain 
are subject to the correctness of the Mathematica computations. 

The idea is to  first apply the MidNormal algorithm using a  Goldberg tetrahedral tiling with
appropriate parameter, and then to project
the resulting mesh vertices towards the level surface $F$. The projection is done using
the gradient of the function $f$ defining the level surface.
The GradNormal algorithm moves each vertex of the MidNormal mesh to the location in $\RR^3$ where the gradient
of the level set function $f$
predicts that the level surface $F$ is located. 
In the case of a linear function it would exactly project each vertex onto the zero level set.
It thus produces a first order approximation of $F$, improving the zeroth-order approximation
given by the MidNormal algorithm.  However the  mesh resulting from the projection process can have
sliver triangles with arbitrarily small angles.  An analysis of these badly behaving triangles
 shows that they result from angles that lie  in  one of
 four triangles that are each adjacent in the mesh to a vertex of valence four.  The
 GradNormal algorithm removes these four
triangles and adds a diagonal to the resulting quadrilateral. We show that this second step eliminates all
sliver triangles and results in a high quality mesh.
The parameter $a= \sqrt2/4$ gives the choice of Goldberg tetrahedron shape 
that achieves the maximal smallest angle for this process. With this choice we prove that 
all angles lie in the interval  $[ 35.2^o, 101.5^o]$ when the mesh
is sufficiently fine.

\begin{algorithm}
\caption{GradNormal Algorithm}
\begin{algorithmic}[1]

\Procedure{GradNormal}{$e,f$}       
\State Input a differentiable function $f: I^3 \to \RR$ with level set $F = f^{-1}(0)$ and a choice of scale $e=1/N$.
\State  Apply the MidNormal algorithm with parameters $\sqrt2/4, e$ to obtain a mesh $M(f,e)$.
\State Compute the gradient $\nabla f$ at the vertices of the mesh $M(f,e)$.\
\State Remove each  vertex of valence 4 and its four adjacent triangles.  Add a diagonal to the resulting
quadrilateral, giving two new triangles in the mesh. 
\State Project  each vertex 
${\bf v}$ to ${\bf v} - f({\bf v}) {\bf \nabla f}/ {\bf || \nabla f ||^2} $. 
\State Output the list of triangles ${\mathcal T'}$.
\EndProcedure

\end{algorithmic}
\end{algorithm}
 
There are two ways to  choose a diagonal in step (3). It turns out that this choice does not affect the resulting angle  bounds   
when $a= \sqrt{2}/4$. 
For that choice the quadrilateral is a square, and either diagonal results in two triangles having four angles equal
to $45^0$ and two equal to $90^0$. To fix a choice, 
we add the diagonal that connects the lower valence adjacent vertices.

We first consider the effect on angles of projecting to a rotated plane.
\begin{lemma} \label{rotationprojection}
Suppose that ${\vec v = (v_1, 1) }, v_1 >0  $ is a vector in the first quadrant of the $xy$-plane 
and that ${\vec w  = (w_1, w_2) } \ne \vec 0$ and subtends an angle $\alpha < \pi$  with  $\vec v$.
Rotate the $xy$-plane around the $x$-axis through an angle of  $ \theta$, $0 \le \theta \le \pi/2$
and denote the orthogonal projections of the rotated vectors 
${\vec v} , {\vec w}$ back to the $xy$-plane by ${\vec v(\theta))} , {\vec w(\theta)}$.
Then as  $\theta$ increases from  $0$ to $\pi/2$ the angle $\alpha(\theta)$ between ${\vec v(\theta)} $ and ${\vec w(\theta)}$ satisfies:\\
(1) If ${\vec w}$ is parallel to the positive $x$-axis or to the negative $y$-axis then $\alpha(\theta)$ is monotonically decreasing.\\
(2)  If ${\vec w}$ is parallel to the negative $x$-axis or to the positive $y$-axis  then $\alpha(\theta)$ is monotonically increasing.\\
(3) If ${\vec w}$  lies in the interior of the second quadrant then $\alpha(\theta)$ is monotonically increasing.\\
(4) If ${\vec w}$  lies in the interior of the fourth quadrant then $\alpha(\theta)$ is monotonically decreasing.\\
(5) If ${\vec w}$  lies in the  interior of the first quadrant then $\alpha(\theta)$ achieves its minimum
at an endpoint of  the interval  $[ 0, \pi/2 ]$.\\
(6) If ${\vec w}$  lies in the  interior of the  third quadrant then $\alpha(\theta)$ achieves its maximum
at  an endpoint  of  the interval  $[ 0, \pi/2 ]$.
\end{lemma}
\begin{proof}
Rotation about the $x$-axis through an angle of  $ \theta$ takes the point $(x, y, 0)$ to  
$(x,   y \cos \theta, y\sin \theta)$.   Thus  ${\vec v(\theta)}  = (v_1,  \cos \theta)$ and 
${\vec w(\theta)}  = (w_1,  w_2 \cos \theta)$.  The angle between each vector and the $x$-axis is decreasing with $\theta$, 
implying the claims in Cases (1) -- (4).

The last two cases needs a more detailed investigation.
In Case (5) each of $w_1, w_2$ is positive, and we can assume that $ w_2 = 1$ by scaling.  
The angle  $\alpha(\theta)$  between  ${\vec v(\theta)} $ and $ {\vec w(\theta)}$  satisfies
\[
\cos \alpha(\theta) =  \frac{v_1w_1 +   \cos^2 \theta }{ \sqrt{( {v_1}^2 + \cos^2 \theta)}\sqrt{({w_1}^2 +  \cos^2 \theta)}} .
\]
For given vectors $\vec v$ and $\vec w$ the   cosine of $ \alpha(\theta) $ has first derivative  
\[
(\cos \alpha(\theta))' =\frac{ \sin\theta \cos\theta (v_1 -   w_1)^2 (v_1 w_1 -  \cos^2{\theta}) }{ (v_1^2 + \cos^2{\theta})^{3/2} (w_1^2 + \cos^2{\theta})^{3/2}} .
\]

A computation shows that the critical points of $\cos \alpha(\theta) $ lie either at the boundary of the interval $[0,\pi/2]$, or in the case where $ v_1 w_1 <1$, 
at an interior point where $\theta =\arccos{\sqrt{v_1w_1}}$.
A further computation shows that the second derivative at the interior critical point is positive, so there is no interior local maximum. 
Thus the cosine of $\alpha$ is maximized at the endpoints of $\theta \in [0,\pi/2]$,
 implying that the angle  $\alpha(\theta)$  is minimized at one of these two endpoints.
 
 For Case (6), where the angle between ${\vec v} $ and $ {\vec w}$ is greater than $\pi/2$, 
 we note that this angle is complementary  to that between ${\vec v} $ and $ {-\vec w}$, which was studied in Case (5).
 Thus a maximum in this case coincides with a minimum in Case (5), and this again occurs at an endpoint of the interval as claimed.
 \end{proof}
 
\begin{corollary} \label{extremeAngles}
Suppose two vectors in $\RR^3$  are orthogonally projected to a family of rotated planes that begins with the plane 
containing them and contains planes rotated about a line through an angle of at most $\pi/2$.
If the vectors subtend an angle smaller or equal to $\pi/2$
 then the minimum  angle between the projected edges occurs at 
either the initial or final projection. If they subtend an angle  greater than $\pi/2$
then the maximum  angle between the projected edges occurs at 
either the initial or final projection.
\end{corollary} 

We now  compute bounds on the angles produced by the GradNormal algorithm.
The choice of the parameter $a$ affects the resulting angles.
It emerges from a computation that a  Goldberg tetrahedron different from that in the MidNormal algorithm  gives  
optimum angles in the projected GradNormal mesh.
For a given choice of $a$, define $ \theta_{\mbox{min}}(a)$ to
be the greatest lower bound for the angles produced by the GradNormal algorithm using a tiling by tetrahedra of shape $\tau_a$.

\begin{proposition} \label{besta}
 For all $a$, $ \theta_{\mbox{min}}(a)  < 35.42^o$.
\end{proposition}
\begin{proof} 

Angle $\angle ABC$ of $\triangle ABC$ is equal to $\displaystyle  \arccos{ 2 a/\sqrt{4 a^2+1}}$. 
A computation show that this angle is strictly less than $\cos ^{-1}\left(\sqrt{\frac{2}{3}}\right)  \approx 35.2644^o$  for $a > \sqrt{2}/     2$. 
This means that to establish the Proposition, we can restrict attention to $a \le  \sqrt{2}/ 2$,
which corresponds to nonobtuse Goldberg tetrahedra by Lemma~\ref{nonobtuse}.

To get an upper bound on  $ \theta_{\mbox{min}}(a) $ for $a \in (0, \sqrt{2}/ 2 ]$ we first consider two angles that occur in
elementary normal triangles and two planes onto which they could be projected during the algorithm.
Namely angle $\angle BCD$ of $\triangle BCD$ could be projected into the plane containing $ACD$ and 
 angle $\angle ACB$ of $\triangle ABC$ could also be projected to the same plane.
 
 The projected angles as functions of $a$ are 
 \[  \frac{4 a^2+1}{\sqrt{\frac{\left(a^2+1\right) \left(172 a^4+\left(59-36 \sqrt{12 a^2+3}\right) a^2+4\right)}{4 a^2+1}}} \]
and
\[\frac{2-4 a^2}{\sqrt{\frac{\left(a^2+1\right) \left(172 a^4+\left(59-36 \sqrt{12 a^2+3}\right) a^2+4\right)}{4 a^2+1}}}  . \]

The two functions are equal at $a_0=\sqrt{2}/4 \approx 0.35$, as  shown in Figure~\ref{fig:2angles}. 
A lower bound on projected angles must be smaller or equal to the minimum
of these two functions.   For $ a_0 = \sqrt{2}/4$ this minimum equals
\[
 \cos ^{-1}\left(\frac{4}{\sqrt{75-36 \sqrt{2}}}\right) \approx 35.4128^o 
\]
Just from considering these two angles we see that we cannot get all angles greater than $ 35.413^o$.
Thus $ \theta_{\mbox{min}}(a) < 35.42^o$ for all $a$, as claimed.  
\end{proof} 

\begin{figure}[htbp]
  \includegraphics[scale=0.6]{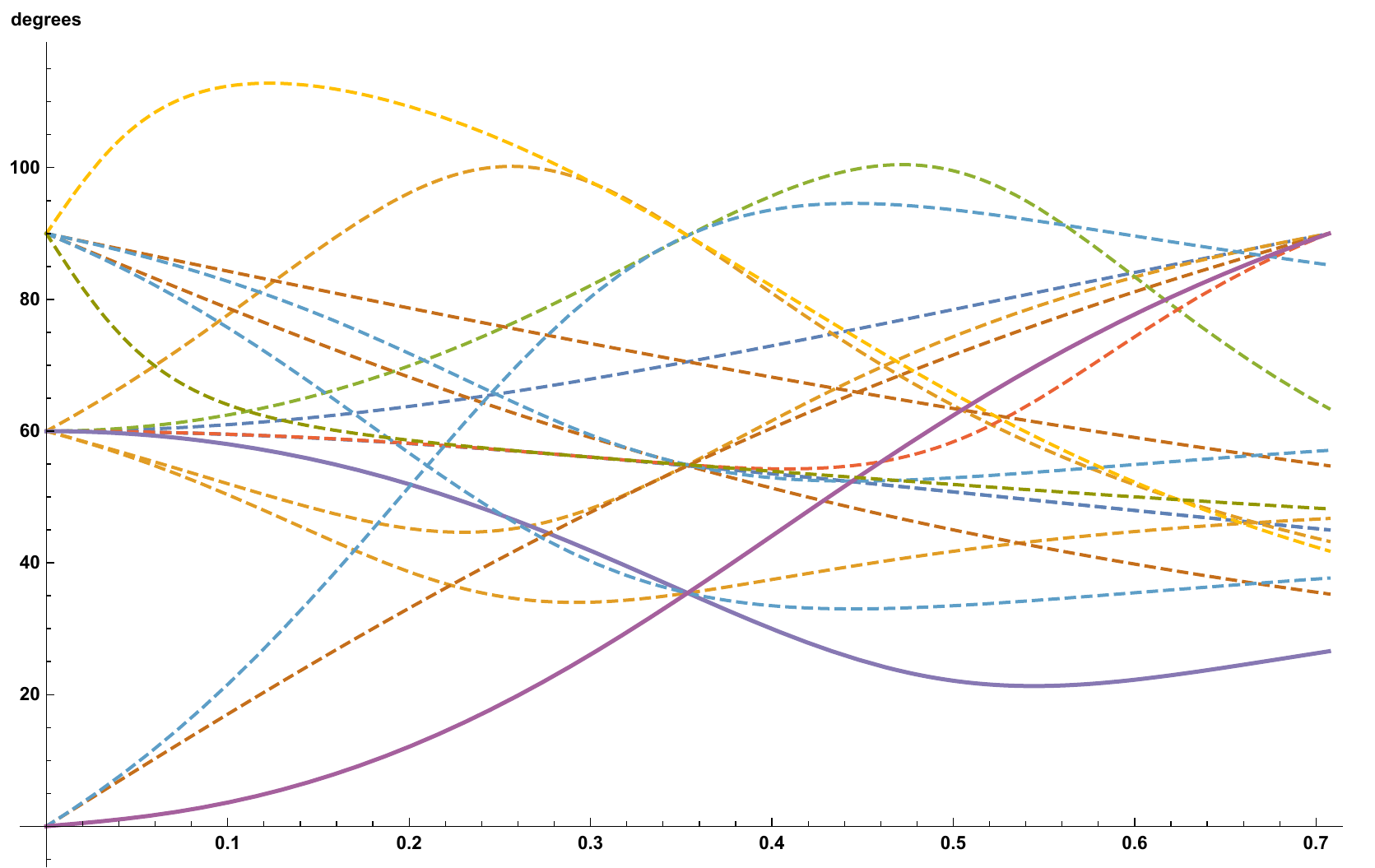}
  \vspace{-0.2cm}
  \caption{The  angles formed by projecting triangles $ABC$ and $BCD$ to faces $ABD, ACD, BCD$ and $ABC, ABD, ACD$ plotted as a function of $a$. The angle of the  projection of $\angle BCD$ and $\angle ACB$  to the plane containing $ACD$ 
  is shown in the two lowest graphs, indicated by undashed curves.
  The minimum of these two projected angles is maximized at  $a_0 = \sqrt{2}/4$ and is $\approx 35.4128^o$.  }
  \label{fig:2angles}
\end{figure}
 
We now show that with this choice of $a = a_0 = \sqrt{2}/4$, the minimum angle  $  \theta_{\mbox{min}}(a_0) >35.2^o$.
Thus the value $a_0= \sqrt{2}/4$ gives a near optimal value for the minimal mesh angle produced by the GradNormal
algorithm. The angles realized for arbitrary $a$ are  bounded in Theorem~\ref{GradNormalThm}.

 \newtheorem*{GradNormalThm}{Theorem~\ref{GradNormalThm}}
\begin{GradNormalThm}
Let $F=f^{-1}(0) \in \RR^3$  be a compact level surface of a smooth function $f$ with regular value 0. 
 For $a= \sqrt{2}/4$ and $e \to 0$\\
(1) The nearest point projection to $F$ gives a  homeomorphism from the  triangular mesh  $M^1(f,e)$ to  $F$.\\
(2)  The surface  $M^1(f,e)$ piecewise-$C^1$ converges to $F$. \\
(3) The mesh angles lie in the interval $[ 35.2^o, 101.5^o]$.
\end{GradNormalThm}
\begin{proof} 
 Let $ M(f,e)$  be the  mesh produced by the MidNormal algorithm with $a= \sqrt{2}/4$ 
 and $M^1(f,e)$ a projection of  $M(f,e)$  along gradient vectors of $f$ 
towards the surface $F$ as in the GradNormal algorithm. 
Since $F$ is smooth and compact it has bounded curvature and as $e \to 0$,
its intersection with a tetrahedron $\tau$ is increasingly closely approximated by a plane.
This plane can be chosen to be a tangent plane of $F$, but for our purposes
we choose it to be a plane that intersects the edges of the tetrahedron
at points where $F$ intersects these edges.
When the surface $F$ separates the vertices of $\tau$ so as to
define an elementary normal disk $E$, then the plane $Q$ separating the same vertices and
intersection the edges of $\tau$ at points where $F$ intersects these edges 
smoothly converges to $F$ on  a neighborhood of $\tau$
of radius $e$.  Thus the angles of the nearest point projection of an elementary normal triangle in   $\tau_a$
of diameter less than $e$ onto $F$
gives angles that converge as $e \to 0$
 to the angles determined by the nearest point projection onto the plane $Q$.
 
We note that in the GradNormal projection we don't  project vertices onto the surface $F$, but rather onto the
plane where $F$ would be if $f$ was a linear function.  This plane smoothly converges to $F$ in a neighborhood of $\tau_{a,e}$
as $e \to 0$.
We conclude that in computing the angles of a projection of an elementary normal triangle in $\tau_{a,e}$
whose three points have been projected to $F$, we can assume, with arbitrarily small error as $e \to 0$,
that $F$ is a plane that separates the vertices of $\tau$ in the same way as the normal surface $F$.

We now classify the various cases of how a plane $F$ can intersect  a tetrahedron $\tau_a$. 
There are four cases where$F \cap \tau_a$ is a triangle and three where it is a quadrilateral that
is divided into two triangles along a diagonal.  An additional case occurs when
four adjacent tetrahedra meet along an edge of valence four and produce a rhombus which is 
divided into two triangles.  Counting cases, we see that there are altogether 12 triangles and 36 angles that 
can be projected onto some plane.

The case valence-4 vertex in $ M(f,e)$ requires special treatment and we consider it first.
Such vertices come from intersection with an edge of length $c$ in a Goldberg tetrahedron, as in Figure \ref{GoldbergTet}.  

\noindent
{\bf Case of a valence-4 vertex in $ M(f,e)$:} 
This vertex appears in the mesh  when four  elementary normal triangles meet the edge $AD$ of length $c$
at its midpoint $M$.
This edge has a dihedral angle of $90^o$ in each of the four adjacent tetrahedra, and 
the four adjacent tetrahedra combine to form an octahedron as in Figure~\ref{fig:octahedron}. 

\begin{figure}[htbp]
  \includegraphics[scale=0.3]{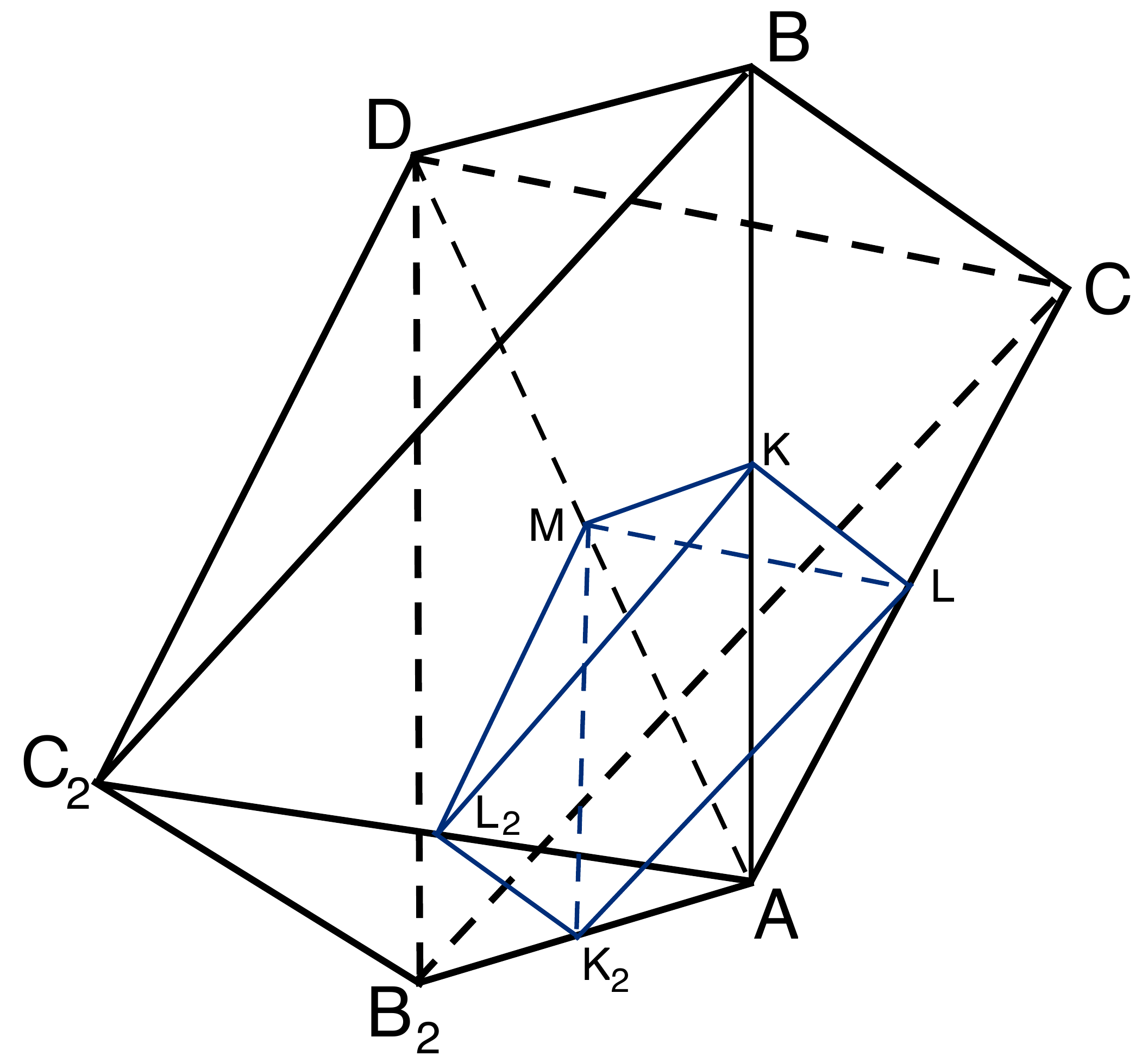}
  \vspace{-0.2cm}
  \caption{Four adjacent tetrahedra meet along $AD$, forming an octahedron. The  mesh surface meets this  octahedron in four triangles, 
  with a common valence-4 vertex at $M$.}
  \label{fig:octahedron}
\end{figure}
 
We consider first the case where $F$ is a plane that intersects the octahedron separating vertex $A$ from vertices $B,C,D$.  
We denote by $X$ the closure of the set of unit vectors perpendicular to such planes, oriented to point towards
$A$.  
We denote by $Y$ the subset  of $X$ consisting of normals to planes separating vertex $A$ from vertices $B,C,D,B_2,C_2$.

For a plane separating vertex $A$ from vertices $B,C,D,B_2,C_2$, the induced mesh has a valence-4 vertex where it intersects edge $AD$.
The GradNormal algorithm removes the four triangles adjacent to the edge $AD$: $\triangle KLM$, $\triangle KL_2M$, $\triangle K_2LM$ and $\triangle K_2L_2M$. 
Note that the four vertices  $ B, C, B_2 ,C_2$ are coplanar,
since there is a reflection through $M$ preserving the octahedron and interchanging $A$ and $D$,  
$B$ and $B_2 $,  and $C$ and $C_2 $. 
These four triangles form a pyramid  $MKLK_2L_2$ whose base is a flat rhombus parallel to rhombus $BCB_2C_2$. 
For $a= \sqrt{2}/4$, the rhombus is a square that realizes dihedral angles
 of $45^o$ with the faces $ABC$, $ABC_2$, $AB_2C$ and $AB_2C_2$ of the octahedron, as indicated in Figure~\ref{fig:octahedron}.
We now analyze the location of the set $Y$ in the unit sphere.

\begin{claim}\label{sphericalquad} 
Suppose $F$ is a plane  separating vertex $A$ from vertices $B,C,D,B_2,C_2$.
Then the unit normal vector of the plane $F$
lies in the interior of a spherical  quadrilateral  $Y \subset X$. The vertices of $Y$ are 
 normal to the faces $ABC$, $ABC_2$, $AB_2C$ and $AB_2C_2$. 
\end{claim}
\begin{proof}
The set of planes with these separation properties
is a subset of the 3-dimensional set of planes in $\RR^3$,
and their unit normal vectors $Y$ form a 2-dimensional subset of the
unit sphere.  If a plane with normal vector in $Y$ does not meet a vertex of
the octahedron then it is in  the interior of an open disk contained in $Y$,
since it can be rotated in any direction while remaining in $Y$.  The same
is true for planes that meet only one vertex  of the octahedron, since
they too can be rotated in all directions while still passing through only this vertex.
Planes in $Y$ meeting two vertices of the octahedron can be rotated only in one circular direction,
and lie along a geodesic arc on the 2-sphere that forms part of $\partial Y$.  Planes that
meet three or more vertices of the octahedron cannot be rotated while maintaining
their intersection with these points, and thus form vertices of $\partial Y$. To understand 
$Y$ we consider which planes separating vertex $A$ from vertices $B,C,D,B_2,C_2$
meet three or more vertices, giving a vertex of $\partial Y$ on the unit sphere, or meet
two vertices, giving an edge of $\partial Y$.

\begin{figure}[htbp]
    \includegraphics[scale=0.2]{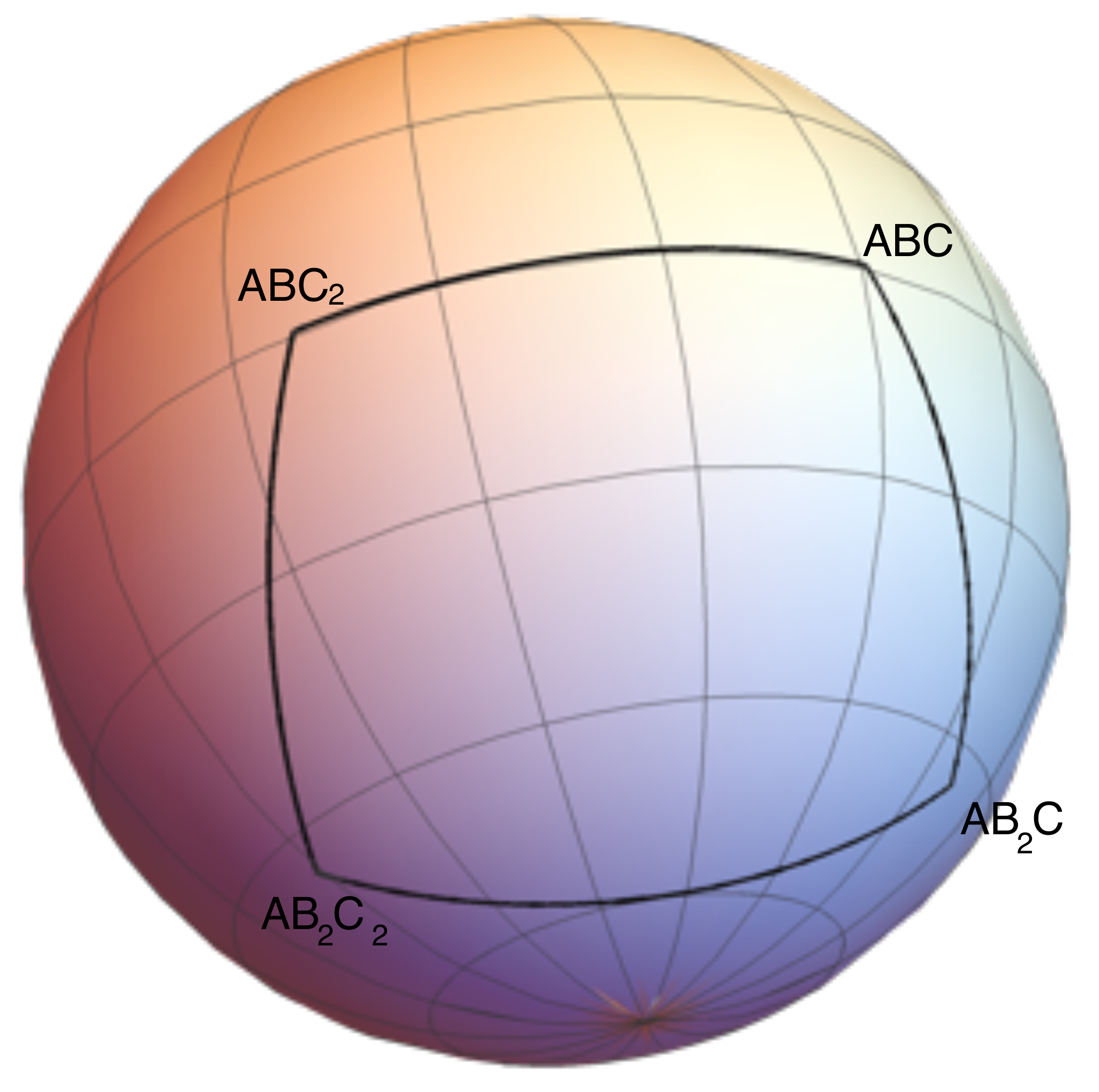}
  \caption{The spherical quadrilateral $Y$ indicates normal directions to planes that separate vertex $A$ from vertices $B,C,D,B_2,C_2$.} 
  \label{fig:Octahedron_SphericalQuad}
\end{figure}

Moreover any plane separating  $A$ from $B,C,D,B_2,C_2$ can be displaced through parallel planes 
towards $A$ till it contains $A$.  It follows that the vertices of $Y$ are determined by triples of vertices that include $A$ and are
limits of planes with the right separation property. These are given by normals $\vec n_{ABC }, \vec n_{ABC_2}\, vec n_{AB _2C }, \vec n_{AB_2C_2}$ to
the faces $ABC$, $ABC_2$, $AB _2C$ and $AB_2C_2$, each of which
gives a vertex of $\partial Y$. These four points on the unit sphere are vertices
of a spherical quadrilateral forming $Y$.  All planes that separate
vertex $A$  from the other vertices of the octahedron with normal pointing towards $A$
have unit normal vectors lying inside $Y$.  See Figure~\ref{fig:Octahedron_SphericalQuad}.
\end{proof}
  
In the GradNormal algorithm we replace the four triangles adjacent to  edge $AD$ with the rhombus $BCB_2C_2$,
divided into two triangles along a diagonal. We need to estimate the angles of these
two triangles after they are projected onto a plane $F$ with normal in the
spherical quadrilateral  $Y$.
Lemma ~\ref{rotationprojection}  implies that the largest and smallest angles 
among projections of the rhombus $KLK_2L_2$ onto $F$
 occur either in the rhombus $KLK_2L_2$ itself or at a plane whose normal lies in $\partial Y$. 
 For $a = \sqrt{2}/4$, this rhombus is a square, and a diagonal divides it into a pair of
 $45^o,45^o,90^o$ triangles. 

We project these two triangles onto  planes with normals  on $\partial Y$. The rhombus $KLK_2L_2$ projects to a parallelogram, so
the two triangles project to congruent triangles, and it suffices to consider the angles of one, say $KLK_2$. We investigate what
angles result from projecting  triangle $KLK_2$ onto a plane normal to $\partial Y$.
Each point in an arc of $\partial Y$ is normal to a plane obtained by rotating one   face of the octahedron to another through an edge containing $A$.
One set of angles results from projecting each of the three angles of triangle $KLK_2$ to planes determined by the spherical arc from $\vec n_{ABC_2}$ to $\vec n_{AB_2C_2}$.
We parameterize an arc of normal vectors $\vec v(t)$ passing from $\vec v(0) = \vec n_{ABC_2}$  to $\vec v(1) = \vec n_{AB_2C_2}$ and
compute the angles resulting from projecting triangle $KLK_2$ to planes  normal to $\vec v(t)$.
These angles are then given by a collection of functions of a parameter $t\in [0,1]$.
The three angle functions from triangle $KLK_2$   are plotted in Figure~\ref{fig:Octahedron_ProjectedAngles}.
The absolute minimum of the  three angle functions on this arc of $\partial Y$ is  $\approx 35.3004^o > 35.25^o$,
and the absolute maximum is  $\approx 101.445^o < 101.45^o$.
We then do a similar computation for each of the other arcs on $\partial Y$.  Figure~\ref{fig:Octahedron_ProjectedAngles2}
shows the angles resulting from projecting  $\triangle KLK_2$  onto the boundary arc of $Y$
 running between $\vec n_{ABC}$  and $\vec n_{ABC_2}$. 
 Again each curve lies above $35.25^o$ and below $101.45^o$,
 showing that all projected angles are between these two bounds.
 The remaining two boundary arcs give the same angle functions, due to a symmetry of the octahedron. 

\begin{figure}[htbp]
\centering
\begin{subfigure}{.33\textwidth}
  \centering
  \includegraphics[width=.9\linewidth]{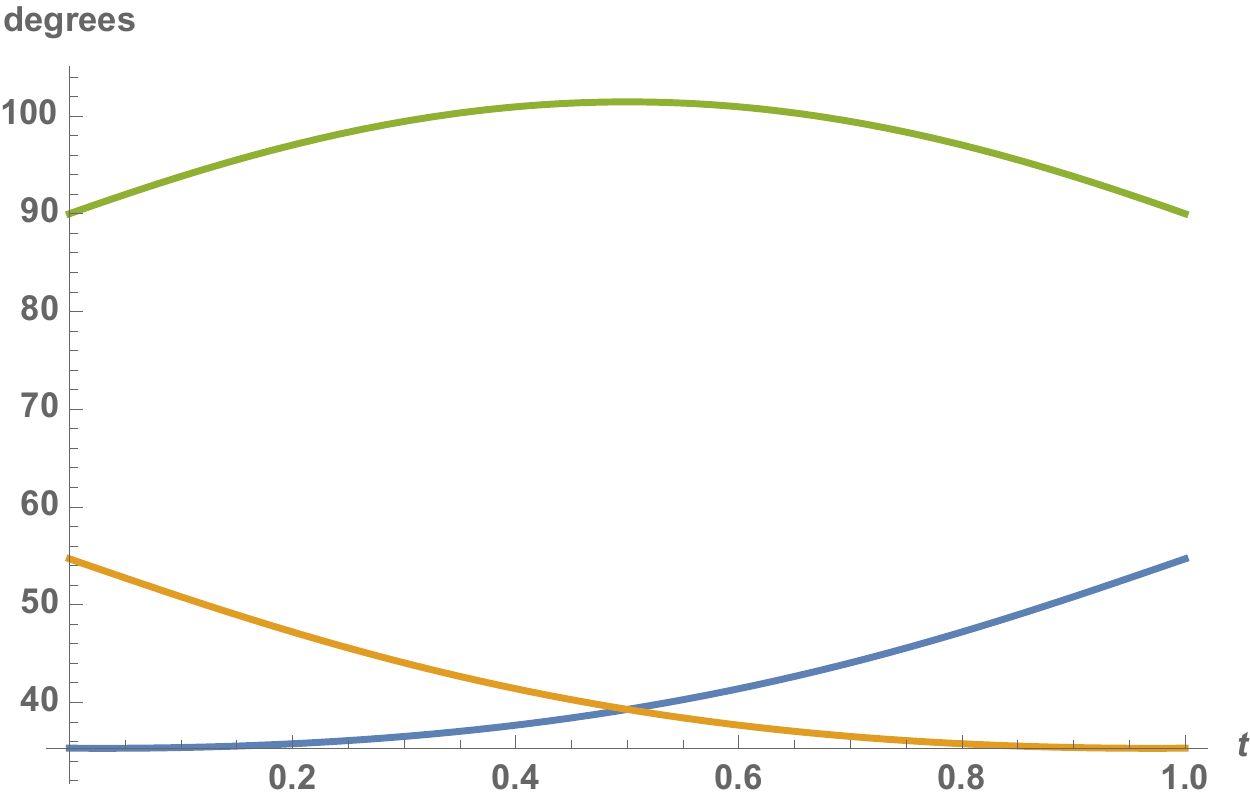}
  \caption{}
\end{subfigure}%
\begin{subfigure}{.33\textwidth}
  \centering
  \includegraphics[width=.9\linewidth]{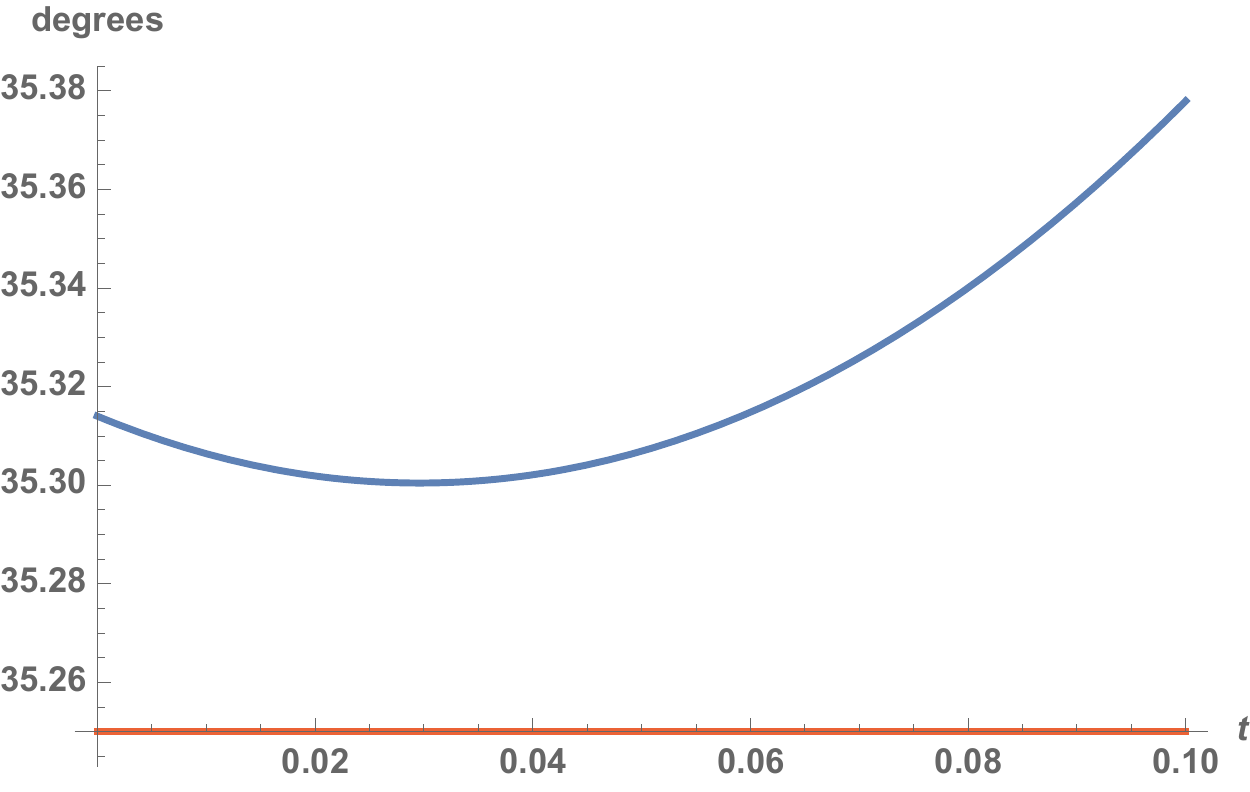}
  \caption{}
\end{subfigure}
\begin{subfigure}{.33\textwidth}
  \centering
  \includegraphics[width=.9\linewidth]{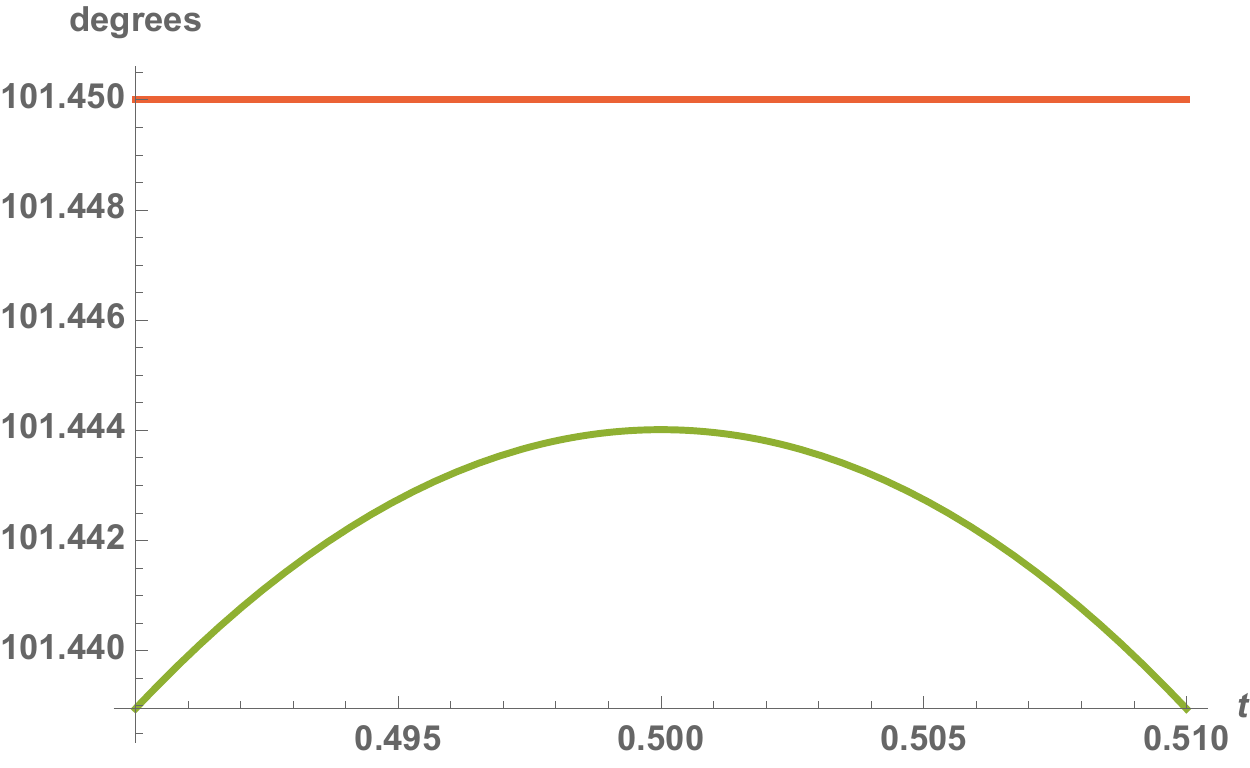}
  \caption{}
\end{subfigure}
\caption{(A) Angles of $\triangle KLK_2$ after projection onto the boundary arc
  from $\vec n_{ABC}$  to $\vec n_{AB_2C}$ of $\partial Y$, parametrized by $t \in [0,1]$. Detailed views of these graphs near (B) $t=0$ and (C) $t=0.5$
  indicate that each curve lies above $35.25^o$ and below $101.45^o$.}
\label{fig:Octahedron_ProjectedAngles}
\end{figure}

\begin{figure}[htbp]
    \includegraphics[scale=0.3]{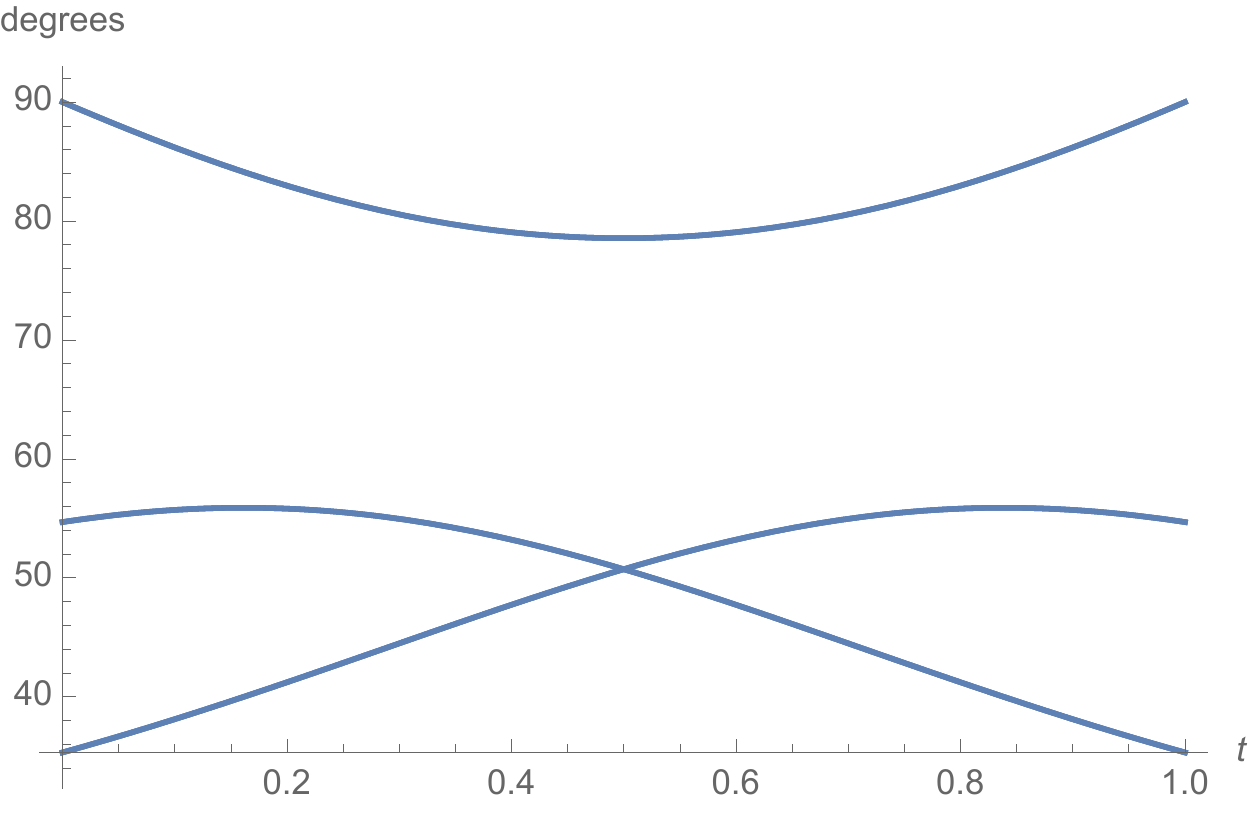}
  \caption{Angles of $\triangle KLK_2$  after projection onto the boundary arc
  from $\vec n_{ABC}$  to $\vec n_{ABC_2}$ of $\partial Y$. Again each curve lies above $35.25^o$ and below $101.45^o$.} 
  \label{fig:Octahedron_ProjectedAngles2}
\end{figure}

We conclude that all projections of the triangles obtained from the diagonally divided rhombus in the GradNormal algorithm
have angles between $35.25^o$ and $101.45^o$. 
 
There is a symmetric case involving a rhombus where $F$ is a plane that separates  vertex $D$ from $A,B,C$.
A symmetry interchanges  $A$ and $D$, and it follows that
this case gives the same angle bounds.

\noindent
{\bf Other Cases:} 
Four remaining cases to consider involve angles obtained by projecting triangles $\triangle KLM$ and $\triangle KNQ$ with edge lengths $(b/2, b/2, c/2)$, and $\triangle LNP$ and $\triangle MPQ$ with edge lengths $( 3a/2, b/2, c/2)$. 
Six remaining cases involve quadrilaterals divided into pairs of triangles: $KLPQ$ is divided into triangles $KLQ$ and $LPQ$, $KMPN$ is divided triangles
$KMN$ and $MNP$,
and $LMNQ$ is divided into triangles $LMN$ and $MNQ$.
We consider these in turn.

\noindent
{\bf Case of $\triangle KLM$:} 
We compute the smallest angle that can occur from a projection of $\triangle KLM$ onto a plane $F$ that
cuts off vertex $A$ from the other vertices  of the tetrahedron, and for which
$\triangle KLM$ is an elementary normal disk.  
The closure of the  set  of possible unit normal vectors for  the plane $F$,  oriented to point towards $A$,
belongs to a spherical triangle  $T$. 
Vertices of  $T$  are unit normal vectors  $\vec n_{ABC}, \vec n_{ABD}, \vec n_{ACD}$ 
to the faces  $ABC$,  $ACD$ and  $ABD$.

The dihedral angles between $\triangle KLM$ and its three adjacent faces are either $60^o$ or $90^o$, and
$F$ can be nearly parallel to one of these faces.
A  projection of  $\triangle KLM$ to  a nearly perpendicular plane can return a triangle with angles close to $0$ or $\pi$, 
giving very poor angle bounds.  Fortunately, the elimination of valence-four vertices in the GradNormal algorithm
 resolves this problem.

If the plane $F$ is almost parallel to the face 
$ABC$ and thus nearly perpendicular to $\triangle KLM$, then $F$ cuts off the vertex $A$ from 
the other vertices of octahedron $ABCDB_2C_2$. 
This case results in a valence-four vertex in the MidNormal mesh, 
the case considered in Lemma~\ref{sphericalquad}. 
The GradNormal algorithm removes the vertex  $M$ in this case and thus avoids
 projecting $\triangle KLM$ to a near perpendicular plane.
The same will apply for planes with normals  in a neighborhood of the vertex
 $\vec n_{ABC}$  of $T$. 
 We now investigate exactly how  $T$ is truncated in the unit sphere
 when we eliminate planes for which MidNormal leads to valence-four vertices at $M$

Call a plane {\em allowable} if it separates
vertex $A$ from vertices $B,C,D$.
Denote by $X$ the closure of the set of unit normal vectors to allowable planes,
oriented to point towards $A$.  
Then $X$ forms a spherical  triangle in the unit sphere
with vertices $\vec n_{ABC}, \vec n_{ABD}, \vec n_{ACD}$.  Inside $X$ is a 
 subset $Y \subset X$ corresponding to normals of allowable planes that 
 separate $A$ from the vertices  $ B_2, C_2$ of the octahedron.
 All normals to planes for which MidNormal gives valence-four vertices at $M$ are in $Y$,
 but some of these  are also normal to planes that lead to higher valence vertices
 at $M$.
This leads us to define another subset $Z \subset Y$ whose points are
in the  closure of normals $\vec v$ with the property that
if  the normal to an allowable plane is in $Z$, then any parallel allowable plane separates
$A$ from vertices $B,C,D,B_2,C_2$.   
It can be seen from Figure~\ref{fig:octahedron} that a neighborhood of $\vec n_{ABC}$ in $X$
lies in $Z$, so this set is non-empty.  
We  now determine the precise shapes of $Y$ and $Z \subset Y$ on the sphere,
determining the configuration shown in Figure~\ref{Z}.

We first consider what points lie in $Y$.
Planes normal to vectors in  $Y$ can be moved to a parallel allowable plane
that separates  $A$ from vertices $B_2,C_2, B,C,D$.  
Any such  plane can be pushed 
through parallel planes in $Y$ towards $A$, 
until it hits $A$, since it separates $A$ from the other five
vertices. The boundary of the set of such planes containing $A$ is a spherical
quadrilateral with vertices corresponding to the normals to
the four faces of the octahedron meeting $A$, namely
 $\vec n_{ABC_2}, \vec n_{AB_2C_2}= \vec n_{BCD}, \vec n_{ACB_2}, \vec n_{ABC} .$
Then $Y$ consists of points insider the spherical quadrilateral with these
four vertices, a subset of the spherical triangle $X$.
 
Next we consider what points lie  in  $Z$.
An allowable plane  normal to a vector in  $Z$ 
must separate $ A$ from $ B_2,C_2, B, C, D$.
This plane can be pushed away from $A$ through parallel planes until it first hits 
one or more of the other five vertices.  
It cannot first hit $D$,  as no allowable plane through $D$ 
separates $A$ from $B_2,C_2,B, C$.

This set of vertices that it hits
must include some subset of  $B, C$
since if it hits only one or both of $B_2,C_2$
then a parallel   plane in $X$would not separate 
$A$ from vertices $B_2,C_2, B,C,D$ and thus its normal
would not lie in $Z$.
We consider which sets of three or more vertices  may be reached
by planes in $Z$ when these planes are translated away from $A$ through parallel planes.
These form some of the vertices of the spherical polygon $Z$.  
Note that the four vertices  $B_2,C_2,B, C $ are coplanar, and form one
 plane defining a vertex of $Z$.  Thus this is the only vertex hit by
pushing a plane in $Z$ away from $A$. Other vertices are found by planes in $Z$ that contain $A$
and two or more additional vertices, giving  vertices of $Z$ at
 $\vec n_{ABC}$,  $\vec n_{ABC_2}$, $\vec n_{AB_2C}$ (but not $\vec n_{AB_2C_2}$, a neighborhood of which lies
 in $Y-Z$).
The resulting region $Z \subset Y$ is shown in Figure~\ref{Z}.
It is the interior of the spherical quadrilateral formed by spherical geodesic
 arcs joining the four vertices $\vec n_{ABC}$,  $\vec n_{ABC_2}$, $\vec n_{AB_2C}$,
 $\vec n_{BCB_2C_2}$. 
 
\begin{figure}[htbp]
\centering
 \includegraphics[scale=0.25]{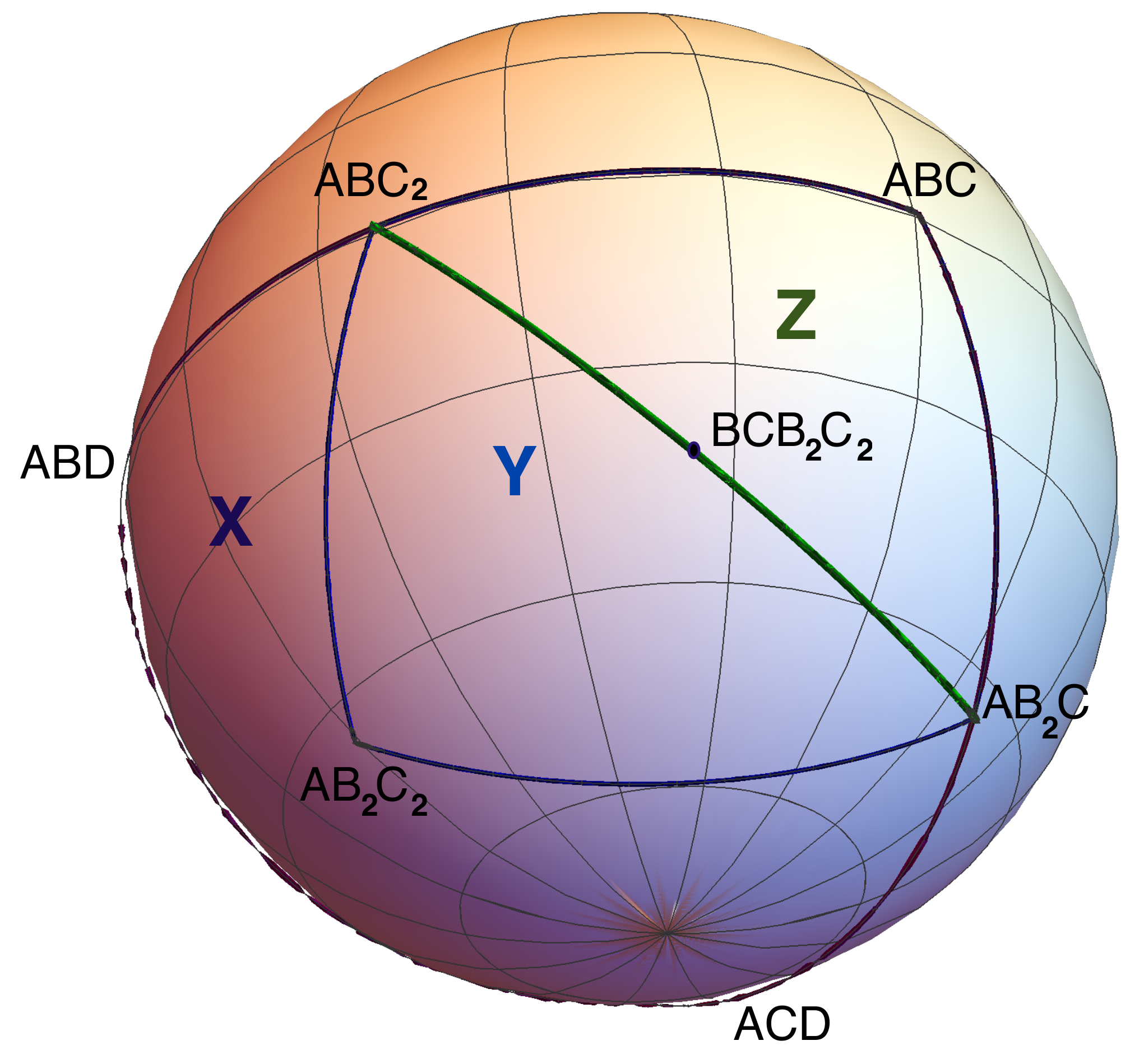}  
   \caption{The spherical triangle $X$ consists of normals to planes separating vertex $A$ from vertices $B,C,D$.
  The region $Y \subset X$ consists of directions for which at least one normal plane gives a vertex of valence-4 at $M$. 
 The region $Z \subset Y$ consists of directions where {\em all} normal planes in $X$  give a   a vertex of valence-4 at $M$.}
  \label{Z}
\end{figure}

The region $ X-Z$ is a spherical quadrilateral, since the vertices
$\vec n_{ABC_2}$,  $\vec n_{BCB_2C_2}$ and $\vec n_{AB_2C}$ lie on a single spherical
geodesic. 
This holds for all $a$ and follows from the fact that lines $BC_2$ and $B_2C$ are parallel to a line of intersection of planes $ABC_2$ and $AB_2C$. Therefore unit normal vectors for planes 
$BCB_2C_2$, $ABC_2$ and $AB_2C$ are coplanar.
Moreover $ X-Z$ is contained in a hemisphere, since all vectors
in $ X$ have positive inner product with $A$.
 
Each  vertex of the  spherical quadrilateral $X-Z$  has distance at most $\pi/2$
from $ \vec n_{KLM} $, as seen by
computing  dihedral angles of the faces of the tetrahedron $\tau_{a_0}$.
The maximum distance of a boundary point from $ \vec n_{BCD} $  occurs at a vertex 
of $ X-Z$, since $ X-Z$ is  a spherical polyhedron contained in a hemisphere.
It follows that each boundary point of  $ X-Z$ 
has distance  at most $\pi/2$ from $\vec n_{KLM} =  \vec n_{KLM} $.
 Corollary~\ref{extremeAngles} implies that extreme angles for the projection of 
$\triangle KLM$ in the GradNormal algorithm are realized either by the triangle itself or 
by a projection to a plane with
normal vector  lying on one of the boundary edges of  $X-Z$. 
There are three angles for $\triangle KLM$ and four boundary edges of  $X-Z$
determining planes onto which they can project. The three angle functions given by  $\triangle KLM$
when projected onto the arc from $\vec n_{ACD}$  to $\vec n_{AB_2C}$ are shown
in Figure~\ref{fig:SphericalQuadTriangleBCD}, as are angles along each of the
other three arcs of $\partial(X-Z)$.

 \begin{figure}[htbp]
\centering
\begin{subfigure}{.24\textwidth}
  \centering
  \includegraphics[width=.9\linewidth]{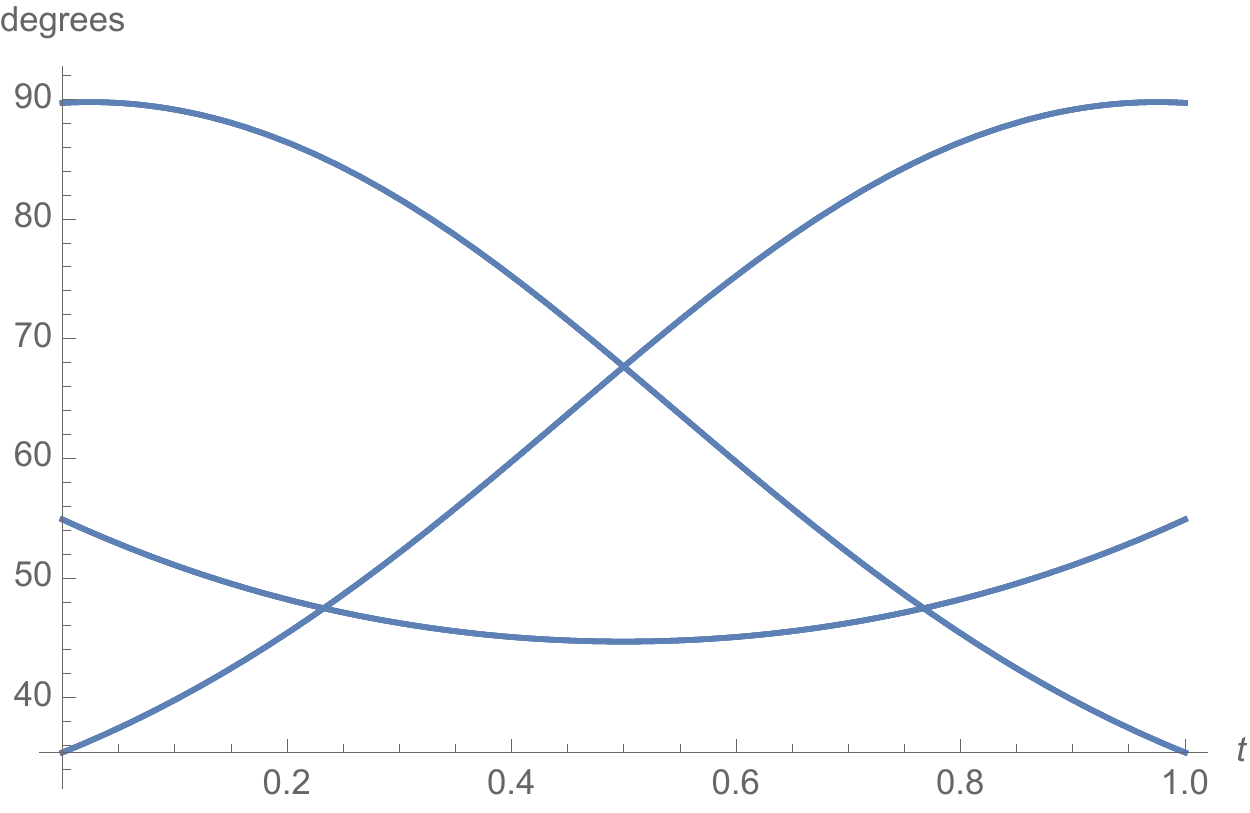}
  \caption{}
\end{subfigure}%
\begin{subfigure}{.24\textwidth}
  \centering
  \includegraphics[width=.9\linewidth]{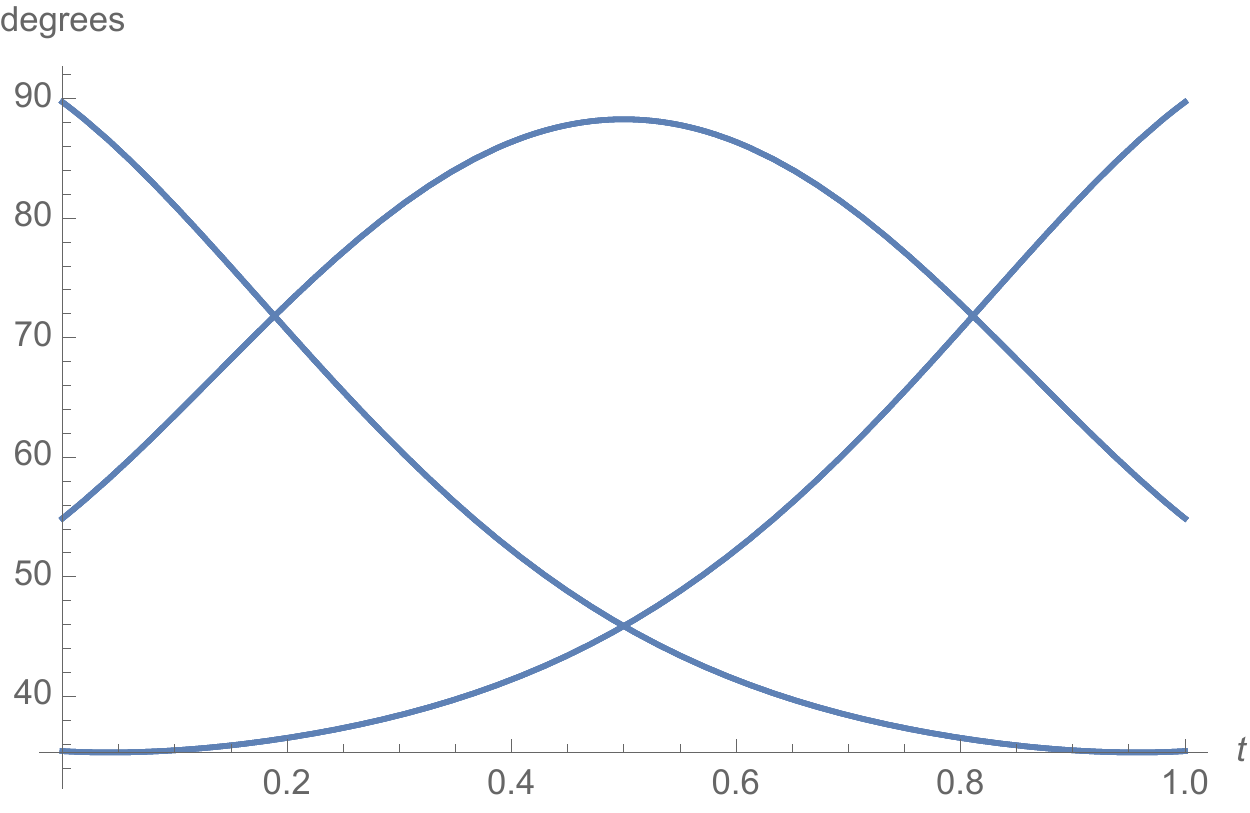}
  \caption{}
\end{subfigure}
\begin{subfigure}{.24\textwidth}
  \centering
  \includegraphics[width=.9\linewidth]{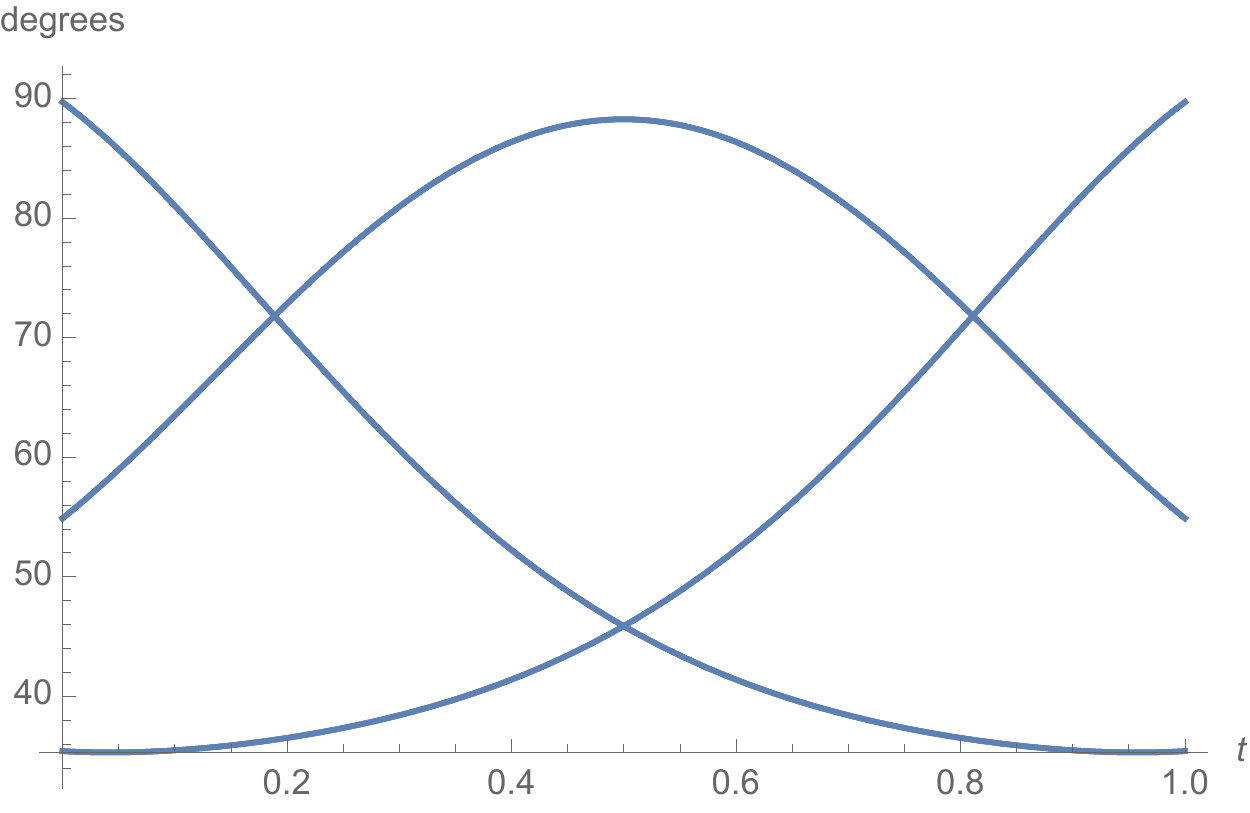}
  \caption{}
\end{subfigure}
\begin{subfigure}{.24\textwidth}
  \centering
  \includegraphics[width=.9\linewidth]{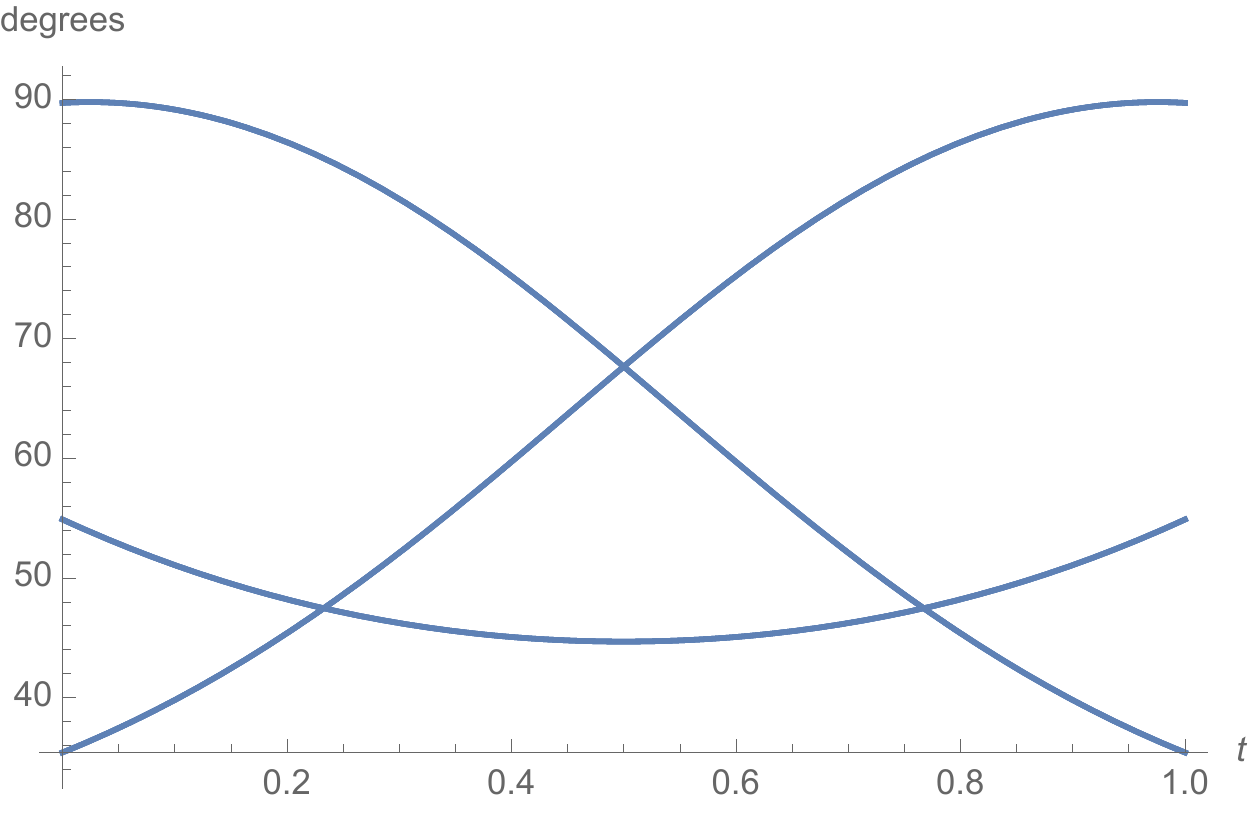}
  \caption{}
\end{subfigure}
\caption{Angles of $\triangle KLM$  after projection onto an arc of $\partial(X-Z)$ running
from (A) $\vec n_{ACD}$  to $\vec n_{AB_2C}$, 
(B)  $\vec n_{BCD}$  to $\vec n_{AB_2C}$, 
(C) $\vec n_{ABD}$  to $\vec n_{ACD}$,
and (D) $\vec n_{ABD}$  to $\vec n_{ABC_2}$. 
Graphs repeat due to symmetries.
Again all angles are in $[35.25^o,~101.45^o]$.} 
\label{fig:SphericalQuadTriangleBCD}
\end{figure}

\noindent
 {\bf Oher Cases:}
Triangles $\triangle KNQ, \triangle LNP$ and $\triangle MPQ$, as well as the triangles
coming from dividing elementary quadrilaterals along a diagonal, all give rise to similar
angle functions for each edge of a corresponding quadrilateral spherical region.
Altogether there are 12 triangles with 36 angles projecting to four edges each, or
144 angle functions in total, each defined on an interval of normal directions connecting
two points on the sphere along a spherical arc.  The union of all these angle functions is
graphed in Figure~\ref{SpherQuads_all72projected_angles}.
 
\begin{figure}[htbp]
    \includegraphics[scale=0.5]{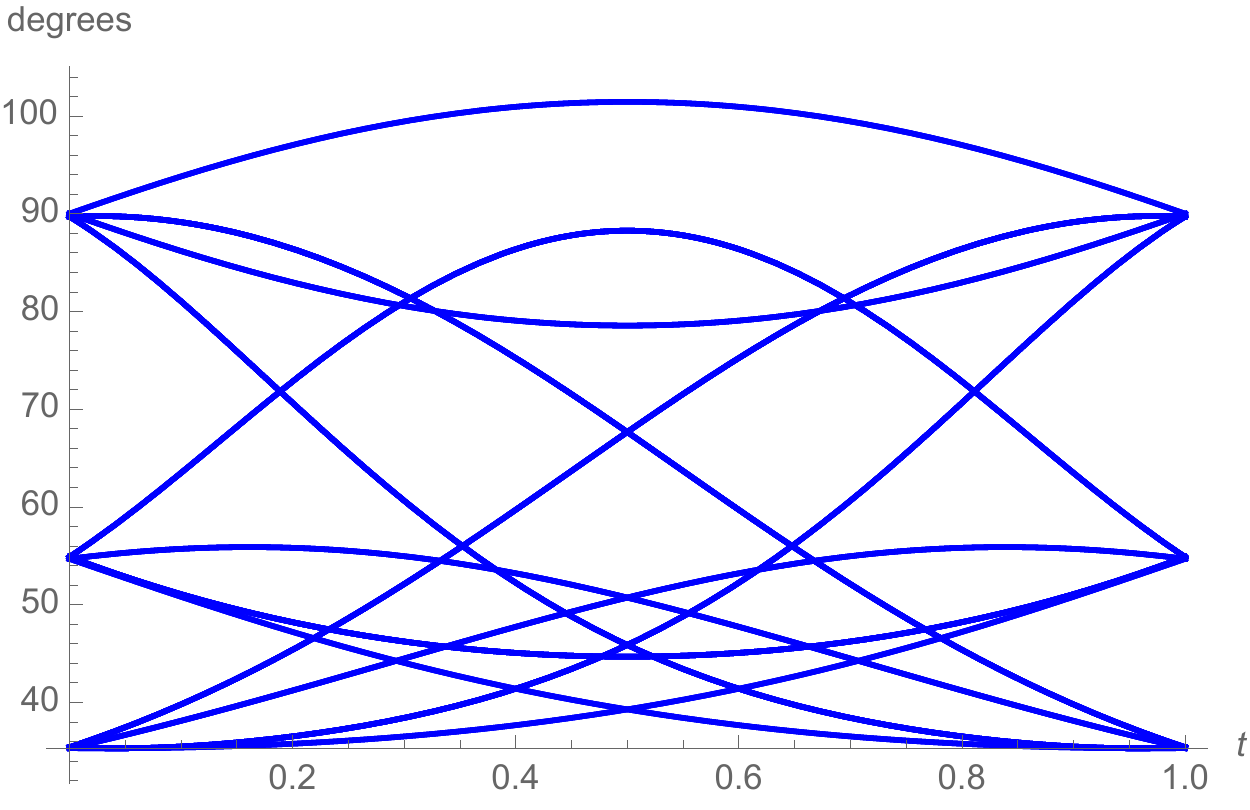}
  \caption{Angles of all triangles in the GradNormal mesh are bounded above and below by the 
  maximum and minimum values obtained in these graphs.  A total of 144 angles are graphed over
  the boundary of spherical regions to produce these functions. Because of symmetries and
  coinciding functions, there are only 12 distinct graphs resulting from these 144 angles.
  All curves lie above $35.25^o$ and below $101.45^o$. } 
  \label{SpherQuads_all72projected_angles}
\end{figure}

We now consider the claims of Theorem~\ref{GradNormalThm}.
In Theorem~\ref{mainTheorem} it was shown that the nearest point projection from mesh
$M(a,e)$ to $F$  is a homeomorphism for $e$ sufficiently small. 
The same argument applies to mesh $M(f,e)$, with the projection given by the
gradient vector for linear $f$ and approximated by the gradient vector for $e$ sufficiently small.
When $f$ is linear, the projected triangle is contained in $F$, and gives a $C^1$ approximation for  
$e$ sufficiently small.  In the argument above, the angle bounds  were established for $F$  a plane,
and also hold for $e$ sufficiently small, since $F\cap \tau$ converges smoothly to the intersection of a plane
with $\tau$ as $e \to 0$.
 \end{proof}

\section{Remarks}
\subsection{Normal surfaces and Marching Tetrahedra}
The simple normal surfaces considered in this paper are similar to the surfaces constructed in the
marching tetrahedra algorithm, though they predate them.  There is a feature of the  general theory of
normal surfaces that gives it the potential to extend the
MidNormal and GradNormal algorithms beyond the settings explored here.
Normal surfaces are well suited for describing surfaces 
that overlap on large subsurfaces.  Many surfaces have this
property, such as  a folded
table cloth, a parachute, the surface of the pages of a book (with many pages touching one another), and the cortical
surface of a brain.  Normal surfaces can efficiently describe such surfaces, and for that  reason
are widely used in computational topology to give efficient representatives of surfaces
in general 3-dimensional manifolds \cite{Haken, HLP}.

\subsection{Other Surface Descriptors}
The MidNormal and GradNormal algorithms 
introduced in this paper  take as input a surface given as a level set of a function on $\RR^3$, 
but they are
amenable to other forms of surface input. For  example, if the input is a poorly triangulated 
surface $F$, then there exist procedures to produce a function on $\RR^3$
that estimates distance from the surface. Producing such a signed distance function 
has been extensively studied in computer graphics \cite{Gib98, PayneToga}.
If the input describing a surface is a point cloud,
methods such as the Moving Least Squares and Adaptive Moving Least Squares procedures 
produce a function giving a level set description of the surface \cite{Dey, ShenOBrienShewchuk}.
This function can then be used as input
to the MidNormal and GradNormal algorithms.

 \subsection{Convergence and curvature}

When we have bounds on the principle curvatures of $F$ we can get angle bounds on the mesh 
for a given value of $e$. We investigate these bounds here, as they are
relevant to whether the GradNormal algorithm can
be used effectively.
The bounds of Theorem~\ref{GradNormal} are guaranteed to apply as the scale size $e \to 0$.  
To test them at a given size, we can fix $e = 1$ and consider how the angle bounds on the mesh
are affected by curvature bounds  on the surface $F$.  Though this can be done rigorously,
we present here some experimental results obtained as a preliminary step.

We set $e=1$ and consider the angles attained by a mesh approximating a surface $F$ 
whose principle curvature are bounded above in absolute value  by a constant $k_0$.
We estimate these angles by modeling $F$ with a sphere.  Since spheres of the appropriate radius have maximal
principal curvature and since they realize all tangent directions, this gives a reasonable approach to modeling the
worst case for an angle bound.  We obtain in this way experimental bounds for the angles obtained in the GradNormal algorithm.
In Table~\ref{angleTable} the result of applying the GradNormal algorithm to spheres of varying radii and tori of revolution at various scales is shown. 
The principle curvatures of the spheres are bounded above by $k_M$,
and the resulting  minimal angles  \textbf{$\theta_m$} and maximum angles  \textbf{$\theta_M$} are shown. This
can be compared to the predicted limiting angle bounds of  $[35.2^o, 101.5^o]$ as $k_M \rightarrow 0$.
 
\begin{table}[h!] \label{curv}
\begin{center}
\begin{tabular}{l|l|l|l}
         \textbf{Spheres} &   \textbf{$k_M$} & \textbf{$\theta_m$} & \textbf{$\theta_M$}   \\
      \hline
&0.23 &$ 33.0^o$& $102.8^o$\\
&0.09 &$ 34.2^o$& $101.3^o$\\
& 0.05 &$ 35.4^o$& $102.7^o$\\
&0.03 &$ 35.2^o$& $101.6^o$\\
\end{tabular}
\ \  \   \  \ 
\begin{tabular}{l|l|l|l}
         \textbf{Tori} &   \textbf{$k_M$} & \textbf{$\theta_m$} & \textbf{$\theta_M$}   \\
      \hline
& 0.5 &$ 9.5^o$& $156.8^o$\\
 &0.2 &$ 32.0^o$& $111.0^o$\\
 &0.1 &$ 38.3^o$& $99.3^o$\\
 &0.05 &$ 33.6^o$& $103.5^o$\\
\end{tabular} 

\vspace{.1in}
\begin{tabular}{l|l|l|l}
         \textbf{Genus 2} &   \textbf{$k_M$} & \textbf{$\theta_m$} & \textbf{$\theta_M$}   \\
      \hline
& 0.57 &$  10.9^o$& $  153.8^o$\\
&0.29 &$   22.2^o$& $  129.0^o$\\
 &0.15 &$  27.8^o$& $ 118.8^o$\\
\end{tabular}

 \end{center}
 \caption{Experimentally attained angle bounds  for spheres  and tori with varying upper bound $k_M$ for the principle curvatures. As $k_M \to 0$ the angles converge to
 the interval $[35.2^o, 101.5^o]$.}
 \label{angleTable}
\end{table}


 \subsection{Further improvements}
It is likely that additional improvements in the angle bounds can be achieved by 
processes such as moving the vertices of the mesh in directions  tangent
to the surface, adding additional vertices, and  performing Delaunay flips. 

\bibliographystyle{plainurl}
\bibliography{MidN}        

\begin{thebibliography}{10}

\bibitem{BernEppstein}
M.~Bern and D.~Eppstein.
\newblock Mesh generation and optimal triangulation.
\newblock In Ding-Zhu Du and Frank Kwang-Ming Hwang, editors, {\em Computing in
  Euclidean Geometry}, volume~1 of {\em Lecture Notes Series on Computing},
  pages 23--90. World Scientific, 1992.

\bibitem{BernEppsteinYao}
M.~Bern, D.~Eppstein, and F.~Yao.
\newblock The expected extremes in a delaunay triangulation.
\newblock {\em International Journal of Computational Geometry \&
  Applications}, 1:79--91, 1991.

\bibitem{BuragoZalgaller}
Y.D. Burago and V.A. Zalgaller.
\newblock Polyhedral embedding of a net.
\newblock {\em Vestnik St. Petersburg Univ. Math.}, pages 66--80, 1960.

\bibitem{ChengDey}
S.W. Cheng, T.K. Dey, and J.~Shewchuk.
\newblock {\em Delaunay Mesh Generation}.
\newblock CRC, 2012.

\bibitem{Chew93}
L.~P. Chew.
\newblock Guaranteed quality triangular meshes.
\newblock In {\em Proceedings of the Ninth annual Symposium on Computational
  geometry}, pages 274--280, San Diego, 1993.

\bibitem{MeshLab}
P.~Cignoni, M.~Callier, M.~Corsini, M.~Dellepiane, F.~Ganovelli, and
  G.~Ranzuglia.
\newblock Meshlab: an open-source mesh processing tool.
\newblock In {\em Sixth Eurographics Italian Chapter Conference}, pages
  129--136, 2008.

\bibitem{deVerdiereMarin}
Y.~Colin de~Verdiere and A.~Marin.
\newblock Triangulations presque equilaterales des surfaces.
\newblock {\em J. Differential Geom}, 32:199--207, 1990.

\bibitem{Dey}
T.K. Dey.
\newblock Curve and surface reconstruction: Algorithms with mathematical
  analysis.
\newblock In {\em Cambridge Monographs on Applied and Computational
  Mathematics}. Cambridge University Press, 2006.

\bibitem{DoiKoide}
A.~Doi and A.~Koide.
\newblock An efficient method of triangulating equi-valued surfaces by using
  tetrahedral cells.
\newblock In {\em IEICE Transactions of Information and Systems}. IEICE, 1991.

\bibitem{EppsteinSullivanUngor}
D.~Eppstein, J.~Sullivan, and A.~Ungor.
\newblock Tiling space and slabs with acute tetrahedra.
\newblock {\em Comp. Geom. Theory $\&$ Applications}, 27(3):237--255, 2004.

\bibitem{Goldberg}
M.~Goldberg.
\newblock Three infinite families of tetrahedral space-fillers.
\newblock {\em J. Comb. Theory}, 16:348--354, 1974.

\bibitem{Luo}
D.~Gu, F.~Luo, and T.~Wu.
\newblock Convergence of discrete conformal geometry and computation of
  uniformization maps.
\newblock {\em Asian Journal of Mathematics}, 23(1):21--34, 2019.

\bibitem{Haken}
W.~Haken.
\newblock Theorie der normalfl\"achen: Ein isotopiekriterium f\"ur den
  kreisknoten.
\newblock {\em Acta Math}, 105:245--375, 1961.

\bibitem{Hass98}
J.~Hass.
\newblock Algorithms for knots and 3-manifolds.
\newblock {\em Chaos, Solitons and Fractals}, 9:569--581, 1998.

\bibitem{HLP}
J.~Hass, J.~Lagarias, and N.~Pippenger.
\newblock The computational complexity of knot and link problems.
\newblock {\em Journal of the ACM}, pages 185--211, 1999.

\bibitem{NormalCode}
J.~Hass and M.~Trnkova.
\newblock Normal mesh files.
\newblock gitlab.com/joelhass/midnormal, 12 2019.

\bibitem{ErtenUngor}
H.Erten and A.~Ungor.
\newblock Computing acute and non-obtuse triangulations.
\newblock In {\em CCCG}, Ottawa, Canada, 2007.

\bibitem{Kneser}
H.~Kneser.
\newblock Geschlossene flachen in dreidimensionalen mannigfaltigkeiten.
\newblock {\em Jahresbericht Math. Verein}, 28:248--260, 1929.

\bibitem{LorensenCline}
W.E. Lorensen and H.E. Cline.
\newblock Marching cubes: A high resolution 3d surface construction algorithm.
\newblock In {\em SIGGRAPH Computer Graphics}, volume~21, pages 163--169. ACM,
  1987.

\bibitem{PayneToga}
B.A. Payne and A.W. Toga.
\newblock Distance field manipulation of surface models.
\newblock In {\em Computer Graphics and Applications}, 12, pages 65--71. IEEE,
  1992.

\bibitem{Gib98}
Spencer S, editor.
\newblock {\em Using distance maps for accurate surface representation in
  sampled volumes}. Symposium on Volume Visualization, IEEE, 1998.

\bibitem{Saraf}
S.~Saraf.
\newblock Acute and non-obtuse triangulations of polyhedral surfaces.
\newblock {\em European J. Combin}, 30:833--840, 2009.

\bibitem{Senechal}
M.~Senechal.
\newblock Which tetrahedra fill space?
\newblock {\em Mathematics Magazine}, 54:227--243, 1981.

\bibitem{ShenOBrienShewchuk}
C.~Shen, J.F. O'Brien, and J.R. Shewchuk.
\newblock Interpolating and approximating implicit surfaces from polygon soup.
\newblock In {\em Proceedings of ACM SIGGRAPH}, Proceedings of ACM SIGGRAPH,
  pages 227--243. ACM, ACM Press, 2004.

\bibitem{Vavasis}
S.A. Vavasis.
\newblock Stable finite elements for problems with wild coefficients.
\newblock {\em SIAM J. Numer. Anal}, 33:35--49, 1996.

\bibitem{Zamfirescu}
C.T. Zamfirescu.
\newblock Survey of two-dimensional acute triangulations.
\newblock {\em Discrete Mathematics}, 313:35--49, 2013.

\end{thebibliography}

\end{document}